\documentclass[a4paper,11pt]{article}
\usepackage{tgpagella}
\usepackage[utf8]{inputenc}
\usepackage{xspace}
\usepackage{natbib}
\usepackage{comment}
\usepackage{hyperref}
\usepackage{booktabs} %
\usepackage[ruled]{algorithm2e} %

\usepackage{tikz}
\usepackage{tikz-network}
\usetikzlibrary{patterns}
\usetikzlibrary{decorations.pathreplacing,calligraphy}
\usetikzlibrary{decorations.pathmorphing}
\tikzset{discont/.style={decoration={zigzag,segment length=12pt, amplitude=4pt},decorate}}
\def\discontarrow(#1)(#2)(#3)(#4);{
  \draw[discont,ultra thick] (#2) -- (#3);
  \draw[ultra thick,-latex] (#1) -- (#2)  (#3) -- (#4);
}
\usetikzlibrary{backgrounds}
\usetikzlibrary{arrows}
\usetikzlibrary{plotmarks}
\usetikzlibrary{fadings}
\usepackage{enumitem}
\usepackage{dsfont}
\usepackage{amsmath,amssymb,amsthm}
\usepackage{enumitem}
\usepackage{geometry}
\usepackage{stmaryrd}
\usepackage{float}
\SetAlFnt{\small}
\SetAlCapFnt{\small}
\SetAlCapNameFnt{\small}
\SetAlCapHSkip{0pt}
\IncMargin{-\parindent}

\newcommand{\calA}{\mathcal{A}}
\newcommand{\calG}{\mathcal{G}}
\newcommand{\calP}{\mathcal{P}}

\newcommand{\RR}{\mathbb{R}}
\newcommand{\NN}{\mathbb{N}}
\newcommand{\EE}{\mathbb{E}}
\newcommand{\PP}{\mathbb{P}}
\newcommand{\ie}{i.e.,\xspace}
\newcommand{\eg}{e.g.,\xspace}

\usepackage{thmtools} 
\newtheorem{theorem}{Theorem}[section]
\newtheorem{lemma}[theorem]{Lemma}

\newtheorem{corollary}[theorem]{Corollary}

\begin{document}

	\title{Improved Bounds for Single-Nomination\\ Impartial Selection}
	 
	\author{Javier Cembrano
	\thanks{Institut für Mathematik, Technische Universität Berlin, Germany}
	\and Felix Fischer			
	\thanks{School of Mathematical Sciences, Queen Mary University of London, UK}
	\and Max Klimm
	\thanks{Institut für Mathematik, Technische Universität Berlin, Germany}
	}
\date{\vspace{-1em}}
\maketitle

\begin{abstract}
We give new bounds for the single-nomination model of impartial selection, a problem proposed by Holzman and Moulin (Econometrica, 2013). A selection mechanism, which may be randomized, selects one individual from a group of $n$ based on nominations among members of the group; a mechanism is impartial if the selection of an individual is independent of nominations cast by that individual, and $\alpha$-optimal if under any circumstance the expected number of nominations received by the selected individual is at least $\alpha$ times that received by any individual. In a many-nominations model, where individuals may cast an arbitrary number of nominations, the so-called permutation mechanism is $1/2$-optimal, and this is best possible. In the single-nomination model, where each individual casts exactly one nomination, the permutation mechanism does better and prior to this work was known to be $67/108$-optimal but no better than $2/3$-optimal. We show that it is in fact $2/3$-optimal for all $n$. This result is obtained via tight bounds on the performance of the mechanism for graphs with maximum degree $\Delta$, for any $\Delta$, which we prove using an adversarial argument. We then show that the permutation mechanism is not best possible; indeed, by combining the permutation mechanism, another mechanism called plurality with runner-up, and some new ideas, $2105/3147$-optimality can be achieved for all $n$. We finally give new upper bounds on $\alpha$ for any $\alpha$-optimal impartial mechanism. They improve on the existing upper bounds for all $n\geq 7$ and imply that no impartial mechanism can be better than $76/105$-optimal for all $n$; they do not preclude the existence of a $(3/4-\varepsilon)$-optimal impartial mechanism for arbitrary $\varepsilon>0$ if $n$ is large.
\end{abstract}

\section{Introduction}

Group decision-making often involves the selection of one of the members of a group based on nominations among members. An early and somewhat unpleasant example of a procedure of this kind is reported from the period of the Athenian democracy in the 5th century BC, when citizens could vote in so-called \emph{ostracisms} for other citizens to be temporarily expelled from the city. After having collected votes engraved into pottery shards from participating citizens, the city would expel the citizen receiving the highest number of votes, subject to a quota. To this date important decisions such as the selection of representatives, speakers, or chairpersons, or the award of prizes are made in the same manner, with pottery shards replaced by slips of paper or clicks on a website.

The economic efficiency of such selection procedures of course relies crucially on impartiality of the votes. However, in decisions with such far-reaching consequences as the possible expulsion from one's home, a voter's interest would certainly lie primarily in the fact whether they themselves are selected, and would dominate any concern for the exact identity of another voter selected instead. Such an interweaving of honest opinions regarding other voters eligible for selection with personal interest is problematic since it may prevent voters from revealing their opinions impartially and truthfully. Specifically, in Classical Athens, a citizen fearing expulsion may have been tempted to vote for another citizen likely to also receive a high number of votes, in an attempt to decrease their own risk of expulsion.

In order to decouple opinions regarding eligibility for selection from selfish interest, it makes sense to study \emph{impartial} selection mechanisms in which the probability of selecting any voter is independent of that voter's vote~\citep{holzman2013impartial,alon2011sum}. Voters and votes can conveniently be represented by vertices and edges of a directed graph. A selection mechanism then chooses a vertex for any graph, and is impartial if the selection of any particular vertex is independent of the outgoing edges of that vertex. It is easy to see that the most natural mechanism, already used in the ostracisms, of selecting a voter who receives a maximum number of votes is not impartial. For example, in a situation where two voters receive a maximum number of votes and vote for each other, one of them can change the selection by instead voting for someone else. 

The requirement of impartiality is obviously appealing, but it turns out to be very demanding and thus incompatible with other important properties. Specifically, any deterministic impartial selection mechanism must in some situation select an individual who receives no votes at all or \emph{not} select an individual receiving votes from everyone else~\citep{holzman2013impartial}. This result was shown for a single-nomination model, \ie for selection mechanisms on graphs where all outdegrees are equal to one. If we measure the quality of a mechanism in terms of the worst ratio over all scenarios between the number of votes for the selected individual and the maximum number of votes for any individual, it follows that for any deterministic impartial mechanism this ratio is at most $1/(n-1)$, where $n$ is the number of individuals. Fortunately this ratio can be improved significantly using randomization~\citep{alon2011sum}, with the best possible ratio lying somewhere in the interval $[67/108,35/48]\approx[0.6204,0.7292]$~\citep{fischer2015optimal}.

\subsection{Our Contribution}

We improve both the lower and upper bound, to show that the best possible performance guarantee for impartial selection in the single-nomination model lies in the interval $[2105/3147,76/105] \approx [0.6689,0.7238]$.

The existing lower bound of $67/108$ due to \citet{fischer2015optimal} is achieved by the so-called permutation mechanism, which considers vertices from left to right along a random permutation and maintains a candidate vertex with maximum indegree from the left.
We provide a tight analysis of the permutation mechanism showing that it achieves a performance guarantee of $\alpha(\Delta)$ in a situation where some individual receives $\Delta$ votes, where $\alpha(\Delta)=(3\Delta+2)/(4\Delta+4)$ if $\Delta$ is even and $\alpha(\Delta)=\alpha(\Delta-1)$ otherwise. This implies a performance guarantee of $2/3$ in any situation. 
The analysis fixes a particular maximum-indegree vertex $v^*$ and allows an adversary to choose the probability with which any vertex other than $v^*$ has indegree at least $t$ from the left, for every $t\in[\Delta]$. The adversary faces a trade-off: on the one hand, large indegrees from the left can be detrimental to the performance of the mechanism because they decrease the probability with which $v^*$ is selected; on the other hand, large indegrees improve the performance guarantee in cases where $v^*$ is not selected. Our tight bound is obtained through a careful analysis that uses the probabilities set by the adversary and shows that possible correlations between indegrees can only improve performance. The proof relies on a coupling argument for the involved probabilities.
The analysis is tight on graphs with a unique maximum-indegree vertex and many vertices with indegree $\lfloor\Delta/2\rfloor$. As the number of the latter increases, the maximum-indegree vertex is selected with probability arbitrarily close to $(\Delta-\lfloor \Delta/2 \rfloor)/(\Delta+1)$, and a vertex with indegree at most $\lfloor\Delta/2\rfloor$ with the remaining probability.

In a many-nominations model, \ie on graphs with arbitrary outdegrees, the permutation mechanism achieves the best performance guarantee of any impartial mechanism, equal to~$1/2$ \citep{fischer2015optimal}. It is natural to ask whether the permutation mechanism is best possible in the single-nomination model as well. We give a negative answer to this question, by constructing a mechanism with performance guarantee greater than $2/3$. When the number $n$ of vertices is small, such a guarantee can be obtained by selecting the out-neighbor of a random vertex. For $n\geq 6$, an improvement on the permutation mechanism is achieved by a convex combination of the permutation mechanism and a modified version of another mechanism from the literature, plurality with runner-up. 

We finally give new upper bounds on the performance guarantee of any impartial mechanism, equal to $(3n^3-19n^2+30n-4)/(4n(n-2)(n-4))$ for graphs with $n\geq 6$ vertices. These bounds are obtained by analyzing sets of linear constraints imposed by impartiality and the requirement to select a vertex on a class of graphs consisting of a $2$-cycle and paths of varying lengths directed at the cycle vertices. They improve on the best upper bounds known previously for all $n\geq 7$, and from $35/48\approx 0.7292$ to $76/105\approx 0.7238$ in the worst case over all $n$.\footnote{The previous bounds were given by \citet{fischer2015optimal} without proof, and claimed to be tight for $n\leq 9$. Comparison with our bounds shows that the claim regarding tightness is incorrect.} Like the previous bounds they tend to $3/4$ as $n$ goes to infinity, and one may conjecture that $3/4$-optimality is achievable in the limit.

\subsection{Related Work}

Impartiality of the kind we discuss here was first considered for the allocation of a divisible good~\citep{de2008impartial}, and then for selection~\citep{alon2011sum,holzman2013impartial}. \citeauthor{holzman2013impartial} showed that in the single-nomination model and for deterministic mechanisms, impartiality is incompatible with the requirement to select any vertex with indegree $n-1$ and never select a vertex with indegree~$0$.
\citeauthor{alon2011sum} studied impartial mechanisms selecting a fixed number~$k$ of vertices in the many-nomination model, showing that deterministic such mechanisms cannot provide a non-trivial performance guarantee for any~$k$ whereas the best randomized one for $k=1$ is between $1/4$- and $1/2$-optimal. The permutation mechanism was introduced and shown to achieve $1/2$-optimality by \citet{fischer2015optimal}. An axiomatic characterization of symmetric randomized mechanisms was provided by \citet{mackenzie2015symmetry}, and improved performance guarantees for different values of~$k$ were given by \citet{bjelde2017impartial}.

\citet{bousquet2014near} proposed a slicing mechanism that achieves a performance guarantee of~$1$ in the limit as $\Delta \to\infty$. \citet{caragiannis2022impartial} pointed out that the slicing mechanism is also $O(n^{8/9})$-additive in the sense that it guarantees a gap of $O(n^{8/9})$ between the indegree of the selected vertex and the maximum indegree. \citeauthor{caragiannis2022impartial} then provided an $O(\sqrt{n})$-additive mechanism for the single-nomination model and an $O(n^{2/3}\ln^{1/3}n)$-additive mechanism for the many-nominations model, both of them randomized.
\citet{cembrano2022impartial-ec} showed $O(\sqrt{n})$-additivity in the single-nomination model can also be achieved deterministically, whereas in the many-nominations model deterministic impartial mechanisms cannot improve on the trivial additive guarantee of $n-1$.

\citet{tamura2014impartial} studied mechanisms selecting a variable number of vertices in the single-nomination model, and specifically a deterministic version of plurality with runner-up that selects a maximum-indegree vertex and possibly a second vertex whose indegree is smaller by at most~$1$. An axiomatic characterization of this mechanism was provided by \citet{tamura2016characterizing}, and \citet{cembrano2022optimal} studied the trade-off between the number of selected vertices and the gap between the maximum degree and the degree of any selected vertex in settings with more than a single nomination.

Following the observation that a highly publicized mechanism for scientific peer review due to \citet{merrifield2009telescope} expressly does not separate honest opinions from personal interest, impartial selection has finally been studied as an element of mechanisms for impartial peer review~\citep[\eg][]{kurokawa2015impartial,aziz2016strategyproof,kahng2018ranking,xu2019strategyproof}.

\section{Preliminaries}

For $n\in \NN$, let $\calG_n$ be the set of directed graphs with $n$ vertices, no loops, and one outgoing edge per vertex, \ie the set
\[
    \left\{(V, E): V = \{1,2,\dots,n\}, E \subseteq (V \times V ) \setminus \bigcup_{v\in V}\{(v,v)\}, |E\cap (\{v\}\times V)|=1 \text{ for every }v\in V\right\}
\]
Let $\calG = \bigcup_{n\in \NN} \calG_n$. 
For $G=(V,E)\in\calG$ and $v\in V$, let $N^+(v, G)=\{u\in V:(v,u)\in E\}$ denote the out-neighborhood of~$v$ in~$G$, $N^-(v, G)=\{u\in V:(u,v)\in E\}$ the in-neighborhood of~$v$ in~$G$, $\delta^-(v,G)=|N^-(v,G)|$ the indegree of~$v$ in~$G$, and $\delta^-_S(v,G)=|\{u\in S: (u,v)\in E\}|$ the indegree of~$v$ from a particular subset $S\subseteq V$ of the vertices.
We refer to the maximum indegree of any vertex in $G$ as $\Delta(G) = \max_{v\in V} \delta^-(v, G)$, and to an arbitrary (fixed) vertex with maximum indegree in $G$ as $v^*(G) \in \arg\max_{v\in V} \delta^-(v, G)$.
When the graph is clear from the context, we will sometimes drop $G$ from the notation and write $N^-(v)$, $\delta^-(v)$, $\delta^-_S(v)$, $\Delta$ and~$v^*$.
We also let $T(G)=\{v\in V: \delta^-(v)=\Delta\}$ denote the set of maximum-indegree vertices of $G$.
For a graph $G=(V,E)\in \calG$ and $v\in V$, we finally denote by $G_{-v}=(V,E\setminus (\{v\}\times V)$ the graph obtained by omitting $v$'s outgoing edge.

A \textit{selection mechanism} consists of a family of functions $f: \calG_n \to [0, 1]^n$ mapping each graph to a probability distribution over its vertices, such that $\sum_{v\in V}f_v(G) = 1$ for every $G=(V,E)\in \calG$. 
In slight abuse of notation, we use $f$ to refer to both the mechanism and individual functions of the family. 
Mechanism $f$ is \textit{impartial} on $\calG'\subseteq \calG$ if on this set of graphs the outgoing edges of a vertex have no influence on its selection, \ie if for every pair of graphs $G = (V, E)$ and $G' = (V, E')$ in $\calG'$ and every $v\in V$, $f_v(G) = f_v(G')$ whenever $G_{-v}=G'_{-v}$. 
Mechanism $f$ is \textit{$\alpha$-optimal} on $\calG' \subseteq \calG$, for $\alpha \leq 1$, if for every graph in $\calG'$ the expected indegree of the vertex selected by $f$ differs from the maximum indegree by a factor of at most $\alpha$, \ie if
\[
    \text{inf}_{\substack{G\in \calG:\\\Delta(G)>0}}\frac{\EE_{v\sim f(G)} [\delta^-(v, G)]}{\Delta(G)} \geq \alpha.
\]
We will sometimes define a selection mechanism in terms of an algorithm that uses randomness and returns a single vertex. The output of the mechanism in the space $[0,1]^n$ is then simply the vector of probabilities with which each vertex is returned: denoting the output of the algorithm on graph $G=(V,E)\in \calG_n$ by $\calA(G)$, the mechanism is the function $f:G\to [0,1]^n$ such that $f_v(G) = \PP[\calA(G)=v]$ for every $v\in V$.

We finally introduce some notation regarding permutations of the set of vertices. We denote by $\calP(n)$ the set of all permutations of $\{1,\ldots,n\}$. For $\pi\in\calP(n)$, we denote by $\pi_{<v}$ the set of vertices to the left of~$v$ in $\pi$, \ie $\pi_{<v}=\{u\in V: u=\pi_i, \text{$v=\pi_j$ for some $i<j$}\}$. For $\pi\in\calP(n)$ and $S\subseteq V$, we write $\pi(S)$ for the restriction of $\pi$ to $S$, such that, for each $k\in\{1,\ldots,|S|\}$, $\pi_k(S)=v$ if and only if $v\in S$ and $|\{u\in S\cap \pi_{<v}\}|=k-1$.

For $\pi\in \calP(n)$ we write $\pi^{R}$ for the reverse permutation of $\pi$, such that $\pi^R_i=\pi_{n+1-i}$ for all $i\in\{1,\ldots,n\}$, and $\pi^{i,j}$ for the permutation obtained from $\pi$ by swapping the elements in position $i$ and $j$, such that for all $i,j,k\in \{1,\ldots,n\}$,
\[
    \pi_k^{i,j} = \left\{ \begin{array}{ll}
             \pi_j & \text{if $k=i$}, \\
             \pi_{i} & \text{if $k=j$}, \\
             \pi_k & \text{otherwise.} \\
             \end{array}
   \right.
\]

\section{A Tight Analysis of the Permutation Mechanism}
\label{sec:lower-bound}

On graphs with arbitrary outdegrees, impartial selection is solved optimally by the so-called permutation mechanism~\citep{fischer2015optimal}. The mechanism, described formally as \autoref{alg:perm},
chooses a permutation $\pi$ of the vertices uniformly at random and then considers vertices ``from left to right'' in the order of that permutation. At any point the mechanism maintains a candidate vertex $v^P$, along with the indegree~$d$ of that vertex from its left; initially the candidate vertex is the first vertex in the permutation, and its indegree from the left is equal to~$0$. When a new vertex is considered, that vertex becomes the new candidate if and only if it has a higher indegree from the left than the current candidate, whereby any edge from the current candidate vertex is ignored. The vertex selected by the mechanism is the vertex that is the candidate when all vertices have been considered.
\begin{algorithm}[t]
\SetAlgoNoLine
\KwIn{graph $G=(V,E)\in \calG_n$}
\KwOut{vertex $v^P\in V$}
Sample a permutation $\pi=(\pi_1,\ldots,\pi_n)$ of the vertices $V$ uniformly at random\;
initialize $v^P \xleftarrow{} \pi_1$ and $d\xleftarrow{} 0$\;
\For{$j\in \{2,\ldots,n\}$}{
    denote $v\xleftarrow{} \pi_j$ and $\pi_{<v} = \{\pi_i: i\in \{1,\ldots,j-1\}\}$\;
    \If{$\delta^-_{\pi_{<v}\setminus \{v^P\}}(v)\geq d$}{
        update $v^P \xleftarrow{} v$ and $d\xleftarrow{} \delta^-_{\pi_{<v}}(v)$\;
    }
}
{\bf return} $v^P$
\caption{Permutation mechanism ($\mathsf{Perm}$)}
\label{alg:perm}
\end{algorithm}

It is easy to see that the permutation mechanism is impartial on~$\calG$, and indeed also on graphs with arbitrary outdegrees. On graphs with arbitrary outdegrees it is $1/2$-optimal, which is best possible for any impartial mechanism. On~$\calG$ the performance of the mechanism is significantly more difficult to analyze, and it was previously shown to be at least $67/108$-optimal but no better than $2/3$-optimal~\citep{fischer2015optimal}.
We proceed to give a tight analysis of the mechanism on~$\calG$. We specifically analyze the mechanism on graphs with maximum degree~$\Delta$, for any $\Delta$. The following is our main result.
\begin{theorem}
\label{thm:lb-permutation}
  Let $\Delta\in\NN$. Then the permutation mechanism is $\alpha(\Delta)$-optimal on $\{G\in\calG:\Delta(G)=\Delta\}$, where 
    \[
        \alpha(\Delta) = \left\{ \begin{array}{ll} 
        1 & \text{if $\Delta=1$,} \\[1ex]
        \frac{3\Delta+2}{4\Delta+4} & \text{if $\Delta$ is even,} \\[1ex]
        \alpha(\Delta-1) & \text{otherwise.}
        \end{array}
        \right.
    \]
\end{theorem}

Since $\alpha(1)=1$, $\alpha(2)=2/3$, and $\alpha(\Delta)$ is non-decreasing for $\Delta\geq 2$, a performance guarantee of~$2/3$ on $\calG$ follows immediately.
\begin{corollary}\label{cor:lb-2/3}
  The permutation mechanism is $2/3$-optimal on $\calG$.
\end{corollary}

We need some notation. For a graph $G=(V,E)$, let $\pi(G)$ be a uniform random permutation of $V$, 
$v^P(G)$ the vertex selected by the permutation mechanism, and $X(G)=\delta^-(v^P(G))$ the indegree of that vertex. In slight abuse of notation, we will use the same notation for specific realizations of the random variables as well as the random variables themselves, and omit the dependence on $G$ when it is clear from the context.

Our proof of \autoref{thm:lb-permutation} relies on two lemmas. The first is a simple observation about the permutation mechanism that was first made by \citet{bousquet2014near}: the selected vertex has maximum indegree from the left for the permutation used by the mechanism.
\begin{lemma}[\citealp{bousquet2014near}]
\label{lem:max-indegree-left}
    For any graph $G\in\calG$, $v^P\in\arg\max_{v\in V}\{ \delta_{\pi_{<v}}(v)\}$.
\end{lemma}

The second, technical, lemma establishes a weak form of negative correlation between the indegree from the left of an arbitrary maximum-indegree vertex and that of any other vertex. It is key to our improvement over the best previous analysis of the permutation mechanism.

\citeauthor{fischer2015optimal} analyze the permutation mechanism as the limit case of a $k$-partition mechanism, which rather than permuting the set of vertices partitions it into $k$ of sets. Analogously to the permutation mechanism the $k$-partition mechanism then considers sets in a fixed order from left to right and maintains a candidate vertex with maximum indegree from sets to the left. The analysis fixes a particular maximum-indegree vertex $v^*$ and allows an adversary to choose, for any fixed assignment of all vertices except $v^*$ to the sets, the maximum indegree from the left for each set. Vertex $v^*$ itself is assigned to one of the sets uniformly at random. The adversary then faces a tradeoff between choosing large indegrees to stop $v^*$ from being selected when it happens to be in a later set, and choosing small indegrees to reduce the performance of the mechanism when $v^*$ happens to be in an earlier set and may not be selected. For graphs with arbitrary outdegrees the $k$-partition mechanism is $((k-1)/2k)$-optimal, and $1/2$-optimality of the permutation mechanism follows by taking $k$ to infinity. In graphs where all outdegrees are positive, which includes the single-nomination case, the existence of a cycle in the graph and thus of an edge from left to right in any permutation of the vertices, the permutation mechanism is $2/3$-optimal on graphs $G$ with $\Delta(G)\in\{1,2\}$. For graphs $G$ with $\Delta(G)\geq 3$ the adversarial analysis shows $67/108$-optimality.

A weakness of the analysis of \citeauthor{fischer2015optimal}, which prevents it from going beyond $67/108$-optimality, is that it allows the adversary to choose indegrees from the left separately for every assignment of the vertices in $V\setminus\{v^*\}$ to the $k$ sets, or separately for every permutation of these vertices in the case of the permutation mechanism. Our analysis overcomes this limitation by allowing the adversary to choose the probability with which any vertex different from $v^*$ receives at least $i$ edges from its left, for all $i\in\{0,1,\ldots,\Delta\}$. We can then compute the expected indegree of the vertex selected by the permutation mechanism by conditioning on the indegree from the left of $v^*$. The advantage of this approach is that the probabilities are clearly not tied to a specific permutation of $V\setminus\{v^*\}$, which allows the adversarial logic to be applied globally. 

Analogously to the existing analysis, a higher probability that some vertex other than $v^*$ has a large indegree from the left allows such a vertex to be chosen when $v^*$ itself has a large indegree from the left, which is detrimental to the performance of the mechanism, but causes a vertex with relatively large indegree to be chosen even when $v^*$ has a small indegree from the left, which improves the performance of the mechanism. What makes this argument nontrivial is that the events where $v^*$ has a large indegree from the left and where some other vertex has a large indegree from the left may be correlated. 
The following lemma establishes that any such correlation is weakly negative, which allows the argument to go through and will be key to our proof of \autoref{thm:lb-permutation}. 
Specifically, the lemma states that conditioning on $v^*$ having indegree $j<i$ rather than $i$ from the left can only increase the probability that some other vertex has indegree at least $i$ from the left.

\begin{lemma}
\label{lem:correlation-neighbors}
    Let $G=(V,E)\in \calG$. Then, for every $i\in \{1,\ldots,\Delta\}$ and $j\in \{0,\ldots,i-1\}$,
    \[
        \PP\left[\bigcup_{v\in V\setminus \{v^*\}}\left[\delta^-_{\pi_{<v}}(v)\geq i \right] \;\Bigg|\; \delta^-_{\pi_{<v^*}}(v^*)=j\right] \geq \PP\left[\bigcup_{v\in V\setminus \{v^*\}}\left[\delta^-_{\pi_{<v}}(v)\geq i \right] \;\Bigg|\; \delta^-_{\pi_{<v^*}}(v^*)=i\right].
    \]
\end{lemma}
\begin{proof}
Let $G$, $i$, and $j$ be as in the statement of the lemma.
Then
\begin{align*}
    \PP\left[\bigcup_{v\in V\setminus \{v^*\}}\left[\delta^-_{\pi_{<v}}(v)\geq i \right] \;\Bigg|\; \delta^-_{\pi_{<v^*}}(v^*)=j\right] & = \frac{\PP\left[ \delta^-_{\pi_{<v^*}}(v^*)=j \wedge \exists ~ v\in V\setminus \{v^*\}: \delta^-_{\pi_{<v}}(v)\geq i\right]}{\PP\left[\delta^-_{\pi_{<v^*}}(v^*)=j \right]}
    \intertext{and}
    \PP\left[\bigcup_{v\in V\setminus \{v^*\}}\left[\delta^-_{\pi_{<v}}(v)\geq i \right] \;\Bigg|\; \delta^-_{\pi_{<v^*}}(v^*)=i\right] & = \frac{\PP\left[\delta^-_{\pi_{<v^*}}(v^*)=i \wedge \exists ~ v\in V\setminus \{v^*\}: \delta^-_{\pi_{<v}}(v)\geq i \right]}{\PP\left[\delta^-_{\pi_{<v^*}}(v^*)=i \right]}.
\end{align*}
Since $\PP\left[ \delta^-_{\pi_{<v^*}}(v^*)=i \right] = \PP\left[ \delta^-_{\pi_{<v^*}}(v^*)=j \right] = 1/(\Delta+1)$, it suffices to show that the inequality holds for the numerators.
Let
\begin{align*}
    \calP_{ij}(n) & = \left\{ \pi\in \calP(n): \delta^-_{\pi_{<v^*}}(v^*)=j \wedge \exists ~ v\in V\setminus \{v^*\}: \delta^-_{\pi_{<v}}(v)\geq i\right\} \text{ and} \\
    \calP_{i}(n) & = \left\{ \pi\in \calP(n): \delta^-_{\pi_{<v^*}}(v^*)=i \wedge \exists ~ v\in V\setminus \{v^*\}: \delta^-_{\pi_{<v}}(v)\geq i\right\},
\end{align*}
and observe that since $\pi$ is chosen uniformly at random it is enough to show that $|\calP_{ij}(n)|\geq |\calP_{i}(n)|$. We do so by constructing an injective function $g:\calP_{i}(n) \to \calP_{ij}(n)$.

For each $\pi\in\calP(n)$, let $k(\pi),\ell(\pi)\in \{1,\ldots,n\}$ such that $\pi_{i+1}(\{v^*\}\cup N^-(v^*)) = \pi_{k(\pi)}$ and $\pi_{j+1}(\{v^*\}\cup N^-(v^*)) = \pi_{\ell(\pi)}$, \ie $k(\pi)$ is the position of the $(i+1)$st vertex among $v^*$ and the in-neighbors of $v^*$ in $\pi$ and $\ell(\pi)$ is the position of the $(j+1)$st vertex among $v^*$ and the in-neighbors of $v^*$ in~$\pi$. Now define $g:\calP_{i}(n)\to\calP_{ij}(n)$ such that for all $\pi\in\calP_{i}(n)$, $g(\pi)=\pi^{k(\pi),\ell(\pi)}$.
An illustration of~$g$ is shown in 
\autoref{fig:lem-correlation-neighbors}.

Function~$g$ is clearly injective: by definition of $\pi^{k(\pi),\ell(\pi)}$, $g(\pi)=g(\pi')$ implies $\pi=\pi'$. We proceed to show that the codomain of~$g$ is $\calP_{ij}(n)$. Let $\pi\in\calP_i(n)$, and observe that $v^*=\pi_{k(\pi)}$. It follows that $v^*=(g(\pi))_{\ell(\pi)}=(g(\pi)(\{v^*\}\cup N^-(v^*)))_{j+1}$, since
\[
    |N^-(v^*) \cap \left\{g_1(\pi),\ldots,(g(\pi))_{\ell(\pi)-1}\right\}|  = |N^-(v^*) \cap \left\{\pi_1,\ldots,\pi_{\ell(\pi)-1}\right\}|=j,
\]
and thus $\delta^-_{(g(\pi))_{<v^*}}(v^*) = j$.
On the other hand, let $v\in V\setminus \{v^*\}$ with $\delta^-_{\pi_{<v}}(v)\geq t$, which we know exists because $\pi\in \calP_i$.
If $v=\pi_{\ell(\pi)}$, it follows from $\ell(\pi)<k(\pi)$ that
\[
    (g(\pi))_{<v} = (g(\pi))_{<\pi_{k(\pi)}} = \left(\pi_{<\pi_{k(\pi)}} \setminus \{v\}\right) \cup \{v^*\} \supset \pi_{<\pi_{\ell(\pi)}} = \pi_{<v},
\]
hence $\delta^-_{(g(\pi))_{<v}}(v) \geq \delta^-_{\pi_{<v}}(v) \geq t$.
Otherwise, if $v=\pi_m$ for $m\not\in \{k(\pi),\ell(\pi)\}$, then
\[
    (g(\pi))_{<v} = \left\{ \begin{array}{ll}
             \pi_{<v} &  \text{ if } m < \ell(\pi) \text{ or } m > k(\pi), \\
             \left(\pi_{<v}\cup \{v^*\}\right)\setminus \{\pi_{\ell(\pi)}\} &  \text{ otherwise.}
             \end{array}
   \right.
\]
Since $\pi_{\ell(\pi)}\in N^-(v^*)$, we know that $\pi_{\ell(\pi)}\not\in N^-(v)$ and thus $\delta^-_{(g(\pi))_{<v}}(v) \geq \delta^-_{\pi_{<v}}(v) \geq t$ as well.
This shows that $g(\pi)\in \calP_{ij}(n)$, which completes the proof.
\end{proof}

\begin{figure}[t]
\centering
\begin{tikzpicture}[scale=0.95]

\Vertex[y=2.5, Math, shape=circle, color=black, size=.05]{A}
\Vertex[x=1, y=2.5, Math, shape=circle, color=black, size=.05]{B}
\Vertex[x=2, y=2.5, Math, shape=circle, color=black, size=.05]{C}
\Vertex[x=3, y=2.5, Math, shape=circle, color=white, size=.05]{D}
\Vertex[x=4, y=2.5, Math, shape=circle, color=black, size=.05]{E}
\Vertex[x=5, y=2.5, Math, shape=circle, color=black, size=.05]{F}
\Vertex[x=6, y=2.5, Math, shape=circle, color=white, size=.05]{G}

\Edge[Direct, color=black, lw=1pt, bend=40](A)(G)
\Edge[Direct, color=black, lw=1pt](B)(C)
\Edge[Direct, color=black, lw=1pt, bend=-20](C)(A)
\Edge[Direct, color=black, lw=1pt, bend=-30](D)(G)
\Edge[Direct, color=black, lw=1pt](E)(D)
\Edge[Direct, color=black, lw=1pt](F)(G)
\Edge[Direct, color=black, lw=1pt, bend=-20](G)(C)

\Text[x=2,y=2.2]{$v$}
\Text[x=6,y=2.2]{$v^*$}

\draw [very thick, -Latex](3,1.7) -- (3,0.8);
\Text[x=2.8,y=1.25]{$g$}

\Vertex[Math, shape=circle, color=black, size=.05]{A1}
\Vertex[x=1, Math, shape=circle, color=black, size=.05]{B1}
\Vertex[x=2, Math, shape=circle, color=black, size=.05]{C1}
\Vertex[x=3, Math, shape=circle, color=white, size=.05]{D1}
\Vertex[x=4, Math, shape=circle, color=black, size=.05]{E1}
\Vertex[x=5, Math, shape=circle, color=black, size=.05]{F1}
\Vertex[x=6, Math, shape=circle, color=white, size=.05]{G1}

\Edge[Direct, color=black, lw=1pt, bend=30](A1)(D1)
\Edge[Direct, color=black, lw=1pt](B1)(C1)
\Edge[Direct, color=black, lw=1pt, bend=30](C1)(A1)
\Edge[Direct, color=black, lw=1pt](D1)(C1)
\Edge[Direct, color=black, lw=1pt, bend=-30](E1)(G1)
\Edge[Direct, color=black, lw=1pt, bend=-20](F1)(D1)
\Edge[Direct, color=black, lw=1pt, bend=-50](G1)(D1)

\Text[x=2,y=-.3]{$v$}
\Text[x=3,y=-.3]{$v^*$}

\Vertex[x=7.5, y=2.5, Math, shape=circle, color=black, size=.05]{H}
\Vertex[x=8.5, y=2.5, Math, shape=circle, color=black, size=.05]{I}
\Vertex[x=9.5, y=2.5, Math, shape=circle, color=black, size=.05]{J}
\Vertex[x=10.5, y=2.5, Math, shape=circle, color=white, size=.05]{K}
\Vertex[x=11.5, y=2.5, Math, shape=circle, color=black, size=.05]{L}
\Vertex[x=12.5, y=2.5, Math, shape=circle, color=black, size=.05]{M}
\Vertex[x=13.5, y=2.5, Math, shape=circle, color=white, size=.05]{N}

\Edge[Direct, color=black, lw=1pt, bend=40](H)(N)
\Edge[Direct, color=black, lw=1pt, bend=-20](I)(L)
\Edge[Direct, color=black, lw=1pt](J)(K)
\Edge[Direct, color=black, lw=1pt, bend=-35](K)(N)
\Edge[Direct, color=black, lw=1pt, bend=-20](L)(H)
\Edge[Direct, color=black, lw=1pt](M)(N)
\Edge[Direct, color=black, lw=1pt, bend=-20](N)(L)

\Text[x=11.5,y=2.2]{$v$}
\Text[x=13.5,y=2.2]{$v^*$}

\Vertex[x=7.5, Math, shape=circle, color=black, size=.05]{H1}
\Vertex[x=8.5, Math, shape=circle, color=black, size=.05]{I1}
\Vertex[x=9.5, Math, shape=circle, color=black, size=.05]{J1}
\Vertex[x=10.5, Math, shape=circle, color=white, size=.05]{K1}
\Vertex[x=11.5, Math, shape=circle, color=black, size=.05]{L1}
\Vertex[x=12.5, Math, shape=circle, color=black, size=.05]{M1}
\Vertex[x=13.5, Math, shape=circle, color=white, size=.05]{N1}

\Edge[Direct, color=black, lw=1pt, bend=-20](H1)(K1)
\Edge[Direct, color=black, lw=1pt, bend=30](I1)(L1)
\Edge[Direct, color=black, lw=1pt, bend=40](J1)(N1)
\Edge[Direct, color=black, lw=1pt](K1)(L1)
\Edge[Direct, color=black, lw=1pt, bend=40](L1)(H1)
\Edge[Direct, color=black, lw=1pt, bend=-30](M1)(K1)
\Edge[Direct, color=black, lw=1pt, bend=40](N1)(K1)

\Text[x=11.5,y=-.3]{$v$}
\Text[x=10.5,y=-.3]{$v^*$}

\draw [very thick, -Latex](10.5,1.7) -- (10.5,0.8);
\Text[x=10.3,y=1.25]{$g$}

\end{tikzpicture}
\caption{Example of function $g$ constructed in the proof of \autoref{lem:correlation-neighbors} for $i=3$ and $j=1$. Vertices swapped by $g$ are drawn in white. The indegree from the left of $v^*$ decreases from $i$ to $j$, whereas the indegree from the left of $v$ remains unchanged (left) or increases (right).}
\label{fig:lem-correlation-neighbors}
\end{figure}
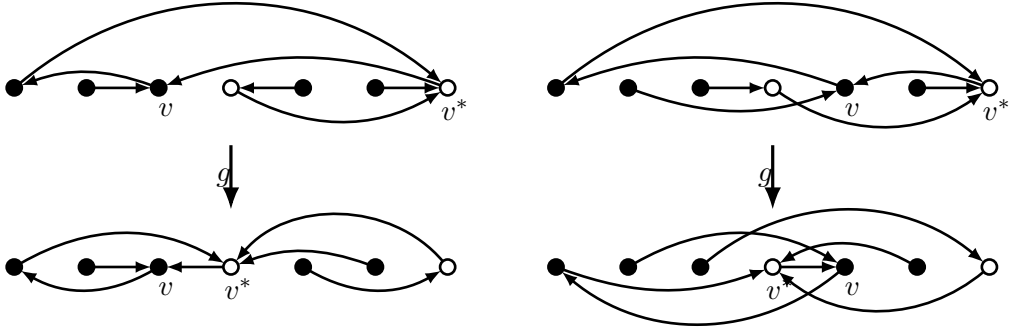
We are now ready to prove \autoref{thm:lb-permutation}.

\begin{proof}[Proof of \autoref{thm:lb-permutation}]
    Let $G=(V,E)\in\calG$.
    If $\Delta=1$, then $\delta^-(v)=1$ for every $v\in V$ and every mechanism is $1$-optimal. We therefore assume $\Delta\geq 2$ in the following.
    
    For $i\in \{0,\ldots,\Delta\}$, let $A_i = \left[\delta^-_{\pi_{<v^*}}(v^*)=i \right]$ be the event that $v^*$ has indegree $i$ from the left, and
    \[
        B_i = \bigcup_{v\in V\setminus \{v^*\}}\left[\delta^-_{\pi_{<v}}(v)\geq i \right]
    \]
    the event that at least one vertex other than $v^*$ has indegree at least~$i$ from the left.

    Using $X$ to denote the indegree of the vertex selected by the permutation mechanism, our goal will be to bound the expectation of $X$ in terms of conditional expectations over a set of disjoint events. We begin by convincing ourselves that the events are indeed disjoint and then give bounds on their probabilities and the corresponding conditional expectations.

    For $i\neq i'$, the events $\left[ A_i \cap \neg B_i\right]$ and $\left[ A_{i'} \cap \neg B_{i'}\right]$ are disjoint since $A_i\cap A_{i'}=\emptyset$. For $j\neq j'$ and all $i$ and $i'$, the events $\left[ A_j \cap B_i \cap \neg B_{i+1}\right]$ and $\left[A_{j'} \cap B_{i'} \cap \neg B_{i'+1}\right]$ are disjoint for the same reason. For $i\neq i'$ and every $j$, the events $\left[A_j \cap B_i \cap \neg B_{i+1}\right]$ and $\left[A_j \cap B_{i'} \cap \neg B_{i'+1}\right]$ are disjoint, since $\left[B_i\cap \neg B_{i+1}\right]$ is the same as $\left[\max_{v\in V\setminus \{v^*\}}\{\delta^-_{\pi_{<v}}(v)\} = i \right]$.     Finally, $[ A_i \cap \neg B_i]$ implies that $\delta^-_{\pi_{<v}}(v)<\delta^-_{\pi_{<v^*}}(v^*)$ for every $v\not=v^*$, and $[ A_{j} \cap B_{i'} \cap \neg B_{i'+1}]$ with $j\leq i'$ implies that $\delta^-_{\pi_{<v}}(v)\geq \delta^-_{\pi_{<v^*}}(v^*)$ for some $v\not=v^*$, so for all $i\in \{1,\ldots,\Delta\}$, $i'\in \{0,\ldots,\Delta\}$, and $j\in \{0,\ldots,i\}$, $[A_i \cap \neg B_i] \cap [A_{j} \cap B_{i'} \cap \neg B_{i'+1}]=\emptyset$.
    
    Now, for all $i\in\{1,\ldots,\Delta\}$, $\PP[A_i]=1/(\Delta+1)$ and thus
    \begin{equation}\label{eq:ex1}
        \PP\left[  A_i \cap \neg B_i \right] = \PP\left[ \neg B_i ~|~ A_i \right]\PP[A_i] = \frac{1}{\Delta+1}\left( 1 - \PP\left[ B_i \mid A_i \right] \right).
    \end{equation}
    Moreover, since $B_{i+1}\subseteq B_i$,
    \begin{equation}\label{eq:ex2}
        \PP\left[ A_j \cap B_i \cap \neg B_{i+1} \right] = \PP \left[ B_i \cap \neg B_{i+1} ~|~ A_j \right] \PP\left[ A_j \right] = \frac{1}{\Delta+1} \left( \PP\left[ B_i ~|~ A_j \right] - \PP\left[ B_{i+1} ~|~ A_j \right] \right).
    \end{equation}
    By \autoref{lem:max-indegree-left}, $X=\Delta$ whenever there exists $i\in\{1,\ldots,\Delta\}$ such that $\delta^-_{\pi_{<v^*}}(v^*)=i$ and $\delta^-_{\pi_{<v}}(v)<i$ for all $v\neq v^*$. Thus, for all $i\in\{1,\ldots,\Delta\}$,
    \begin{equation}\label{eq:ex3}
        \EE\left[X ~|~ A_i \cap \neg B_i\right] = \Delta .
    \end{equation}
   Again by \autoref{lem:max-indegree-left}, $X\geq i$ for $i\in\{0,\ldots,\Delta\}$ if there exists $j\in\{0,\ldots,i\}$ such that (i) $\delta^-_{\pi_{<v^*}}(v^*)=j$, (ii) there is a vertex $v\in V\setminus\{v^*\}$ with $\delta^-_{\pi_{<v}}(v)=i$, and (iii) there is no vertex $v\in V\setminus\{v^*\}$ with $\delta^-_{\pi_{<v}}(v)>i$. 
   Thus, for every $i\in \{0,\ldots,\Delta\}$ and $j\in\{0,\ldots,i\}$, 
   \begin{equation}\label{eq:ex4}
    \EE\left[ X ~|~ A_j \cap B_i \cap \neg B_{i+1}\right] \geq i.
   \end{equation}
    
    We now claim that
    \begin{align*}
        \EE[X] & \geq \sum_{i=1}^{\Delta} \EE\left[ X ~|~ A_i \cap \neg B_i\right] \PP\left[ A_i \cap \neg B_i\right] + \sum_{i=0}^{\Delta} \sum_{j=0}^{i} \EE\left[ X ~|~ A_{j} \cap B_{i} \cap \neg B_{i+1} \right] \PP\left[ A_{j} \cap B_{i} \cap\neg B_{i+1} \right]\\
	& \geq \frac{1}{\Delta+1} \left( \Delta \sum_{i=1}^{\Delta}        \left( 1 - \PP\left[ B_i ~|~ A_i \right] \right) +                   \sum_{i=0}^{\Delta} \sum_{j=0}^{i} i \left( \PP\left[ B_i ~|~ A_j \right] - \PP\left[ B_{i+1} ~|~ A_j \right] \right)            \right)\\
        & =\frac{1}{\Delta+1}\left( \Delta^2 - \Delta \sum_{i=1}^{\Delta}\PP\left[B_i ~|~ A_i \right] + \sum_{i=1}^{\Delta} i \sum_{j=0}^{i} \PP\left[ B_i ~|~ A_j \right] - \sum_{i=2}^{\Delta+1} (i-1)\sum_{j=0}^{i-1} \PP\left[B_i ~|~ A_j  \right] \right)\\
	& =\frac{1}{\Delta+1}\left( \Delta^2 - \Delta \sum_{i=1}^{\Delta}\PP\left[B_i ~|~ A_i \right] + \sum_{i=1}^{\Delta} i \sum_{j=0}^{i} \PP\left[ B_i ~|~ A_j \right] - \sum_{i=1}^{\Delta} (i-1)\sum_{j=0}^{i-1} \PP\left[B_i ~|~ A_j  \right] \right)\\
        & = \frac{1}{\Delta+1}\left( \Delta^2 - \Delta \sum_{i=1}^{\Delta}\PP\left[B_i ~|~ A_i \right] + \sum_{i=1}^{\Delta} \sum_{j=0}^{i-1} \PP\left[ B_i ~|~ A_j \right] + \sum_{i=1}^{\Delta} i \cdot \PP\left[ B_i ~|~ A_i  \right] \right).
    \end{align*}
    Indeed, the first inequality holds because the sum on the right-hand side is over disjoint events, the second inequality by~\eqref{eq:ex1},~\eqref{eq:ex2},~\eqref{eq:ex3}, and~\eqref{eq:ex4}. The first and third equalities follow by re-arranging, the second equality holds because the terms for $i=1$ and $i=\Delta+1$ in the last sum both vanish.
    
    We can now use \autoref{lem:correlation-neighbors}, which states that for every $i\in \{1,\ldots,\Delta\}$ and $j\in \{0,\ldots,i-1\}$ it holds $\PP\left[ B_i ~|~ A_j \right] \geq \PP\left[ B_i ~|~ A_i \right]$, to obtain that
    \begin{align*}
	\EE[X] & \geq \frac{1}{\Delta+1}\left(\Delta^2 - \Delta \sum_{i=1}^{\Delta}\PP\left[B_i ~|~ A_i \right] + \sum_{i=1}^{\Delta} i \cdot \PP\left[ B_i ~|~ A_i \right] + \sum_{i=1}^{\Delta} i \cdot \PP\left[ B_i ~|~ A_i  \right] \right) \\
        & = \frac{1}{\Delta+1}\left( \Delta^2 - \sum_{i=1}^{\Delta}(\Delta-2i)\PP\left[ B_i ~|~ A_i \right] \right).
    \end{align*}
    The sum in the final expression can be bounded as
    \[
        \sum_{i=1}^{\Delta}(\Delta-2i)\PP\left[ B_i ~|~ A_i \right] \leq \sum_{i=1}^{\lfloor \Delta/2 \rfloor}(\Delta-2i)\PP\left[ B_i ~|~ A_i \right] \leq \sum_{i=1}^{\lfloor \Delta/2 \rfloor}(\Delta-2i) = \left\lfloor\frac{\Delta}{2}\right\rfloor \left(\Delta-\left\lfloor \frac{\Delta}{2}\right\rfloor-1\right),
    \]
    where the first inequality follows by dropping non-positive terms and the second inequality because $\PP\left[ B_i ~|~ A_i \right]\leq 1$.
    Thus, if $\Delta$ is even,
    \[
        \frac{\EE[X]}{\Delta} \geq \frac{1}{\Delta(\Delta+1)} \left[ \Delta^2 - \frac{\Delta}{2}\left(\frac{\Delta}{2}-1\right)\right] = \frac{3\Delta^2+2\Delta}{4\Delta(\Delta+1)} = \frac{3\Delta+2}{4\Delta+4} = \alpha(\Delta).
    \]
    If $\Delta$ is odd,
    \[
	\frac{\EE[X]}{\Delta} \geq \frac{1}{\Delta(\Delta+1)} \left[ \Delta^2 - \frac{\Delta-1}{2}\left(\frac{\Delta+1}{2}-1\right)\right] = \frac{3\Delta^2+2\Delta-1}{4\Delta(\Delta+1)} = \frac{3\Delta-1}{4\Delta} = \alpha(\Delta-1).
    \tag*{\raisebox{-.5\baselineskip}{\qedhere}}
    \]
\end{proof}

To prove \autoref{thm:lb-permutation} it was sufficient to condition on the indegree from the left of an arbitrary maximum-indegree vertex. In graphs with multiple maximum-indegree vertices, one would expect the mechanism to only do better. Indeed, on graphs with $k$ maximum-indegree vertices, the permutation mechanism is $k/(k+1)$-optimal. This can be seen by considering the expected indegree from the left of the right-most maximum-indegree vertex.

The bounds in \autoref{thm:lb-permutation} are tight. For $\Delta=2$ this follows from an observation of \citeauthor{fischer2015optimal}, who identified a family of graphs on which the permutation mechanism selects the unique maximum-indegree vertex with a probability approaching $1/3$ as $n$ increases.
The following result generalizes this to graphs with arbitrary maximum degree.

\begin{theorem}
\label{thm:ub-permutation}
Let $\Delta\in\NN$, $\varepsilon>0$. Then there exists a graph $G\in\calG$ such that $\Delta(G)=\Delta$ and the permutation mechanism is \emph{not} $(\alpha(\Delta)+\varepsilon)$-optimal on $G$, where 
    \[
        \alpha(\Delta) = \left\{ \begin{array}{ll} 
        1 & \text{if $\Delta=1$}, \\[1ex]
        \frac{3\Delta+2}{4\Delta+4} & \text{if $\Delta$ is even}, \\[1ex]
        \alpha(\Delta-1) & \text{otherwise.}
        \end{array}
        \right.
    \]
\end{theorem}

The proof of this result can be found in \autoref{app:thm-ub-permutation}.
It relies on graphs with one maximum-indegree vertex and a large number $n'$ of vertices with indegree $\lfloor \Delta/2\rfloor$, so that the permutation mechanism selects the unique maximum-indegree vertex with a probability approaching $(\Delta-\lfloor \Delta/2 \rfloor)/(\Delta+1)$ as $n'$ increases.
An example for $\Delta=4$ is depicted in \autoref{fig:thm-ub-permutation}.

\begin{figure}[t]
        \centering
        \begin{tikzpicture}
        \Vertex[Math, shape=circle, color=black, size=.05, label=n'+4, fontscale=1, position=below, distance=-.08cm]{A}
        \Vertex[x=2, Math, shape=circle, color=black, size=.05, label=n'+5, fontscale=1, position=below, distance=-.08cm]{B}
        \Vertex[y=2, Math, shape=circle, color=black, , size=.05, label=n'+2, fontscale=1, position=above, distance=-.08cm]{C}
        \Vertex[x=2, y=2, Math, shape=circle, color=black, size=.05, label=n'+3, fontscale=1, position=above, distance=-.08cm]{D}
        \Vertex[x=1, y=1, Math, shape=circle, color=black, size=.05, label=1, fontscale=1, position=right, distance=-.08cm]{E}
        
        \Edge[Direct, color=black, lw=1pt, bend=-15](A)(E)
        \Edge[Direct, color=black, lw=1pt](B)(E)
        \Edge[Direct, color=black, lw=1pt](C)(E)
        \Edge[Direct, color=black, lw=1pt](D)(E)
        \Edge[Direct, color=black, lw=1pt, bend=-15](E)(A)

        \Vertex[x=4, Math, shape=circle, color=black, size=.05, label=n'+7, fontscale=1, position=left, distance=-.08cm]{F}
        \Vertex[x=4, y=2, Math, shape=circle, color=black, size=.05, label=n'+6, fontscale=1, position=left, distance=-.08cm]{G}
        \Vertex[x=4, y=1, Math, shape=circle, color=black, size=.05, label=2, fontscale=1, position=left, distance=-.08cm]{H}

        \Edge[Direct, color=black, lw=1pt](F)(H)
        \Edge[Direct, color=black, lw=1pt, bend=-15](G)(H)
        \Edge[Direct, color=black, lw=1pt, bend=-15](H)(G)

        \Vertex[x=6, Math, shape=circle, color=black, size=.05, label=n'+9, fontscale=1, position=left, distance=-.08cm]{F1}
        \Vertex[x=6, y=2, Math, shape=circle, color=black, size=.05, label=n'+8, fontscale=1, position=left, distance=-.08cm]{G1}
        \Vertex[x=6, y=1, Math, shape=circle, color=black, size=.05, label=3, fontscale=1, position=left, distance=-.08cm]{H1}

        \Edge[Direct, color=black, lw=1pt](F1)(H1)
        \Edge[Direct, color=black, lw=1pt, bend=-15](G1)(H1)
        \Edge[Direct, color=black, lw=1pt, bend=-15](H1)(G1)

        \foreach \x in {7,8,9}
        \filldraw [black] (\x,1) circle (0.8pt);

        \Vertex[x=10.5, Math, shape=circle, color=black, size=.05, label=3n'+5, fontscale=1, position=left, distance=-.08cm]{F2}
        \Vertex[x=10.5, y=2, Math, shape=circle, color=black, size=.05, label=3n'+4, fontscale=1, position=left, distance=-.08cm]{G2}
        \Vertex[x=10.5, y=1, Math, shape=circle, color=black, size=.05, label=n'+1, fontscale=1, position=left, distance=-.08cm]{H2}

        \Edge[Direct, color=black, lw=1pt](F2)(H2)
        \Edge[Direct, color=black, lw=1pt, bend=-15](G2)(H2)
        \Edge[Direct, color=black, lw=1pt, bend=-15](H2)(G2)
        \end{tikzpicture}
        \caption{Example of graph $G$ as constructed in the proof of \autoref{thm:ub-permutation}, for $\Delta=4$, with the corresponding vertex labels. For this value of $\Delta$, in order to show, for instance, that the permutation mechanism is not $\alpha$-optimal for $\alpha=2/3+0.1$, it is sufficient to take $n'=14$.} 
        \label{fig:thm-ub-permutation}
    \end{figure}
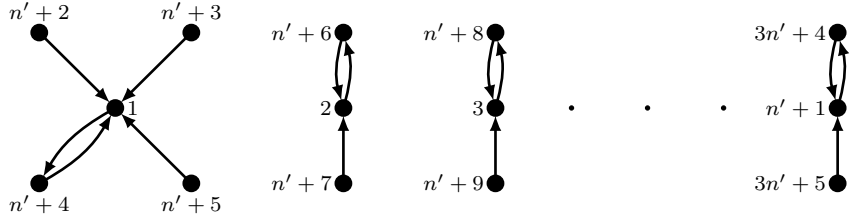

\section{The Permutation Mechanism is Not Best Possible}
\label{sec:improved-lb}

In a model with arbitrary numbers of nominations, the permutation mechanism achieves the best possible performance guarantee of any impartial mechanism, and it is natural to ask whether this could be true in the single-nomination model as well. We have seen in the previous section that the permutation mechanism is $2/3$-optimal on $\calG_n$ for any $n$, and no better than $2/3$-optimal when $n$ can be arbitrarily large. On the other hand, an upper bound of $35/48\approx 0.7292$ was given in prior work on the performance of any impartial mechanism on $\calG$,
and we will later improve this to $76/105\approx 0.7238$. In the limit as $n\to\infty$, the best upper bound is $3/4$. We may therefore ask whether the permutation mechanism is best possible on $\calG$, or on $\calG_n$ as $n\to\infty$. We answer both of these questions in the negative, by showing that a constant improvement over $2/3$-optimality can be achieved for any $n$.
\begin{theorem}
\label{thm:improved-lb-plurality}
    There exists a mechanism that is impartial and $2105/3147$-optimal on $\calG$.
\end{theorem}

We obtain this improved guarantee through a combination of three different impartial mechanisms. One mechanism is used when $n\leq 5$, a fixed convex combination of the other two when $n\geq 6$.

When $n\leq 5$, a simple impartial mechanism introduced by \citet{holzman2013impartial} and called \emph{random dictatorship} improves on $2/3$-optimality.
\begin{algorithm}[t]
\SetAlgoNoLine
\KwIn{graph $G=(V,E)\in \calG_n$}
\KwOut{vertex $v\in V$}
Sample a vertex $u\in V$ uniformly at random\;
{\bf return} $N^+(u)$
\caption{Random dictatorship ($\mathsf{RD}$)}
\label{alg:random-dict}
\end{algorithm}
This mechanism, described in \autoref{alg:random-dict} and denoted by $\mathsf{RD}$, selects a vertex~$v$ uniformly at random and returns the vertex at the other end of the outgoing edge of~$v$. The following lemma, whose proof can be found in \autoref{app:lem-random-dict}, establishes the main properties of this mechanism.
\begin{lemma}
\label{lem:random-dict}
    For $n\in\{2,3,4,5\}$, the random dictatorship mechanism is impartial and $1/2+1/n$-optimal on~$G_n$.
\end{lemma}

For $n\geq 6$, our goal will be to construct a mechanism that performs well in cases where the permutation mechanism is only $2/3$-optimal, and to then randomize between that mechanism and the permutation mechanism.
We have seen in \autoref{thm:lb-permutation} that the class of graphs in which the permutation mechanism is \emph{not} better than $2/3$-optimal only contains graphs with maximum indegree~$2$ or~$3$. The following lemma restricts the class further for a constant improvement over~$2/3$. Its proof can be found in \autoref{app:lem-tight-instances-permutation}.
\begin{lemma}
\label{lem:tight-instances-permutation}
    Let $G\in\calG$ such that the permutation mechanism is \emph{not} $31/45$-optimal on $G$. Then $\Delta(G)\in\{2,3\}$ and $|\{v\in V: \delta^-(v)\geq 2\}|=1$.
\end{lemma}

To improve performance on graphs covered by \autoref{lem:tight-instances-permutation}, we will use a variation of a mechanism called plurality with runner-up, first proposed by \citet{tamura2014impartial}. 
Before we can define it formally, we need to slightly relax the definition of a mechanism to include functions that may not always select, \ie functions $f:\calG_n\to[0,1]^n$ that for every $G=(V,E)\in\calG_n$ satisfy $\sum_{v\in V}f_v(G)\leq 1$. We will call a family of such functions an \textit{inexact mechanism}, and note that our definitions of impartiality and $\alpha$-optimality naturally extend to inexact mechanisms. We will also slightly abuse terminology and refer to an inexact mechanism as a mechanism when the distinction is clear from context.
Plurality with runner-up permutes the vertices uniformly at random and then selects, with probability~$1/2$, the vertex with greatest index among those with maximum indegree.
In addition, if removing the outgoing edge of another vertex would make that vertex the vertex with greatest index among those with maximum indegree, that vertex is also selected with probability $1/2$.

The variation we consider, and which we will call \emph{plurality with runner-up and gap}, improves performance on any graph with a gap of at least~$2$ between the maximum indegree and all other indegrees, and is targeted specially at graphs where one vertex has indegree~$3$ and all other vertices have indegree at most~$1$.
For a given permutation of the vertices, the mechanism selects the vertex with greatest index among those with maximum indegree with probability~$3/4$ or~$1/2$, depending on whether or not removing the outgoing edge of that vertex creates a gap of at least two between its indegree and all other indegrees.
In addition, if removing the outgoing edge of another vertex would make that vertex the vertex with greatest index among those with maximum indegree, that vertex is also selected with probability $1/2$.
Of course, for a fixed permutation, the sum of these selection probabilities could be as large as $5/4>1$. To obtain a well-defined inexact mechanism we thus average over a permutation~$\pi$ chosen uniformly at random and over its reverse~$\pi^R$.
\begin{algorithm}[t]
\SetAlgoNoLine
\KwIn{graph $G=(V,E)\in \calG_n$}
\KwOut{vector $p\in [0,1]^n$ with $\sum_{v\in V}p_v \leq 1$}
    Sample a permutation $\bar{\pi}\in \calP(n)$ uniformly at random\;
    \For{$\pi\in \left\{\bar{\pi}, \bar{\pi}^R\right\}$}{
        initialize $p_v(\pi)= 0$ for every $v\in V$\;
        let $v^F(\pi) = \arg \max_{v\in V} (\delta^-(v,G), \pi_v)$\;
        \eIf{$\delta^-\left(v^F(\pi),G\right) \geq \delta^-\left(v,G_{-v^F(\pi)}\right)+2$ for every $v\in V\setminus \left\{v^F(\pi)\right\}$}{
            \vspace{.1cm}
            update $p_{v^F(\pi)}(\pi) \leftarrow 3/4$
        }{
            update $p_{v^F(\pi)}(\pi) \leftarrow 1/2$
        }
        let $v^S(\pi) = \arg \max_{v\in V\setminus \{v^F(\pi)\}} (\delta^-(v,G), \pi_v)$\;
        \If{$\left(v^S(\pi),v^F(\pi)\right)\in E$ and $\left[\delta^-(v^S(\pi)) = \Delta \text{ or } \left(\delta^-\left(v^S(\pi)\right) = \Delta-1 \text{ and } \pi_{v^S(\pi)} > \pi_{v^F(\pi)}\right)\right]$}{
            \vspace{.1cm}
            update $p_{v^S(\pi)}(\pi) \leftarrow 1/2$
        }
    }
    let $q=\frac{1}{2}\bigl(p(\bar{\pi})+p\bigl(\bar{\pi}^R\bigr)\bigr)$\;
    sample a set $S$ by letting $S=\{v\}$ with probability $q_v$ for every $v\in V$ and $S=\emptyset$ with probability $1-\sum_{v\in V} q_v$ \;
    {\bf return} $S$
\caption{Plurality with runner-up and gap ($\mathsf{PRUG}$)}
\label{alg:pwru-gap}
\end{algorithm}
Plurality with runner-up and gap is described formally as \autoref{alg:pwru-gap}.\footnote{The mechanism can be seen as a straightforward variation of plurality with runner-up through an alternative description: it selects the vertex with maximum indegree with probability $3/4$ rather than $1/2$ whenever all other indegrees are smaller than the maximum indegree by at least~$2$ in a graph where the outgoing edge of the maximum-indegree vertex has been removed; any other vertex is selected with the same probability as in plurality with runner-up. We will use the extensive description as \autoref{alg:pwru-gap} for the results that follow.}
In the description of the mechanism and in its analysis we use regular inequality signs as well as maximum and minimum operators to denote lexicographic comparison of pairs of the form $(\delta^-(v),v)$ for different vertices~$v$.

By construction, plurality with runner-up and gap is $3/4$-optimal on graphs $G$ with $\Delta(G)=3$ and $\delta^-(v)\leq 1$ for all $v\in V\setminus\{v^*\}$, but it is still only $1/2$-optimal on graphs $G$ with $\Delta(G)=2$ and $\delta^-(v)\leq 1$ for all $v\in V\setminus\{v^*\}$. In addition it sometimes does not select a vertex and is therefore not a mechanism. 

We remove both of these shortcomings by a procedure we call \emph{addition of a default vertex}, a technique first used by \citet{holzman2013impartial} that can be applied to an arbitrary inexact mechanism~$f$. The procedure selects a vertex~$\bar{v}$ uniformly at random, removes its outgoing edge, and runs~$f$ on the resulting graph; if~$f$ ends up not selecting a vertex, the new mechanism selects vertex~$\bar{v}$. A formal description of the procedure is given as \autoref{alg:default-vertex}, its output for an inexact mechanism~$f$ and a graph~$G$ will be denoted by $\mathsf{DV}(f,G)$.
\begin{algorithm}[t]
\SetAlgoNoLine
\KwIn{inexact mechanism $f:\calG_n\to [0,1]^n$, graph $G=(V,E)\in \calG_n$}
\KwOut{vertex $v\in V$}
Sample a vertex $\bar{v}\in V$ uniformly at random\;
define $p\in [0,1]^n$ as $p_{\bar{v}} = 1-\sum_{v\in V}f_v(G_{-\bar{v}}),~ p_v=f_v(G_{-\bar{v}})$ for every $v\in V\setminus \{\bar{v}\}$\;
sample a vertex $v$ with probabilities given by $p_v$ for every $v\in V$\;
{\bf return $v$}
\caption{Addition of a default vertex, $\mathsf{DV}(f,G)$ for input $(f,G)$}
\label{alg:default-vertex}
\end{algorithm}

We will be interested specifically in the mechanism $\mathsf{PRUG}^D$ obtained by addition of a default vertex to plurality with runner-up and gap. The following lemma establishes impartiality of this mechanism as well as guarantees for its performance on some graphs of interest.
\begin{lemma}
\label{lem:prugd}
    $\mathsf{PRUG}^D$ is an impartial mechanism on $\calG$. Moreover, for every $n\geq 6$, it is
    \begin{enumerate}[label=(\roman*)]
        \item \label{lem:prugd-i}  $\alpha(\Delta)$-optimal on $\{G\in \calG_n: \Delta(G)=\Delta\}$, where
        \[
            \alpha(\Delta) = \frac{1}{2} + \frac{7\Delta-9}{6\Delta(3\Delta-2)};
        \]
        \item $65/96$-optimal on $\{G=(V,E)\in \calG_n: \Delta(G)=2\}$;\label{lem:prugd-ii}
        \item $13/18$-optimal on $\{G=(V,E)\in \calG_n: \Delta(G)=3, |\{v\in V: \delta^-(v)\geq 2\}| = 1\}$.\label{lem:prugd-iii}
    \end{enumerate}
\end{lemma}

The proof of this lemma, which is given in full in \autoref{app:lem-prugd}, consists of three main steps. First, we show that plurality with runner-up and gap is an impartial inexact mechanism. That it is an inexact mechanism is established by showing that for any $\pi\in\calP(n)$ such that $\sum_{v\in V}p_v(\pi)=5/4$, it follows that $\sum_{v\in V}p_v(\pi^R)=3/4$. Impartiality is proved directly but can also be seen by observing that the mechanism differs from plurality with runner-up only in that it sometimes selects a particular vertex with probability $3/4$ rather than $1/2$, and that the decision to do so is made independently of the outgoing edge of that vertex.
Second, we show that addition of a default vertex preserves impartiality and turns any inexact mechanism into a mechanism. The key observation regarding impartiality is that the outgoing edge of the default vertex is removed and thus has no influence on any additional probability assigned to the default vertex.
The approximation guarantees, finally, are proved by considering the case where a gap of at least~$2$ exists between the maximum indegree and the next-largest indegree in a graph where the outgoing edges of the default vertex and the maximum-indegree vertex have been removed, and the case where there is no such graph. The existence of a gap can only improve the performance of the mechanism, but the general bound in \autoref{lem:prugd-i} can be shown even without taking such an improvement into account. The bounds in \autoref{lem:prugd-ii} and \autoref{lem:prugd-iii} are obtained by exploiting the particular structure of the graphs in question. For \autoref{lem:prugd-ii}, the crucial step is to bound from above the number of realizations, in terms of default vertex and permutation, for which the maximum-indegree vertex is selected with probability zero; for \autoref{lem:prugd-iii}, it is to bound from below the number of realizations for which the maximum-indegree vertex is selected with probability $3/4$.

To achieve the performance guarantee claimed in \autoref{thm:improved-lb-plurality}, we will do the following: for graphs with at most~$5$ vertices, we run the random dictatorship mechanism; otherwise we run the permutation mechanism with probability~$825/1049$ and $\mathsf{PRUG}^D$ with probability $224/1049$.
\begin{algorithm}[t]
\SetAlgoNoLine
\KwIn{graph $G=(V,E)\in \calG_n$}
\KwOut{probability distribution $[0,1]^n$}
\eIf{$n\leq 5$}{
    {\bf return $\mathsf{RD}(G)$}
}{
    let $Y\sim \text{Bernoulli}(825/1049)$, $y$ a realization of $Y$\;
    \eIf{$y=1$}{
        {\bf return $\mathsf{Perm}(G)$}
    }{
    {\bf return $\mathsf{PRUG}^D(G)$}
    }
}
\caption{Mixture mechanism ($\mathsf{Mix}$)}
\label{alg:improved-lb-plurality}
\end{algorithm}
The resulting mechanism, which we call the mixture mechanism, is given as \autoref{alg:improved-lb-plurality}.
We proceed with the proof of the theorem.
\begin{proof}[Proof of \autoref{thm:improved-lb-plurality}]
    We claim the result for the mixture mechanism.
    Let $n\in\NN$, graphs $G=(V,E),~G'=(V,E')\in \calG_n$, and $v\in V$ with $G_{-v} = G'_{-v}$. If $n\leq 5$, then
    \[
        \mathsf{Mix}_v(G) = \mathsf{RD}_v(G) = \mathsf{RD}_v(G') = \mathsf{Mix}_v(G'),
    \]
    where the second equality holds by impartiality of the random dictatorship mechanism, as stated in \autoref{lem:random-dict}.
    Otherwise
    \begin{align*}
        \mathsf{Mix}_v(G) & = \frac{825}{1049}\mathsf{Perm}_v(G) + \frac{224}{1049} \mathsf{PRUG}^D_v(G)\\
        & = \frac{825}{1049}\mathsf{Perm}_v(G') + \frac{224}{1049} \mathsf{PRUG}^D_v(G')\\
        & = \mathsf{Mix}_v(G'),
    \end{align*}
    where the second equality holds by impartiality of the permutation mechanism and of $\mathsf{PRUG}^D$, established respectively in \autoref{thm:lb-permutation} and \autoref{lem:prugd}.
    We conclude that the mixture mechanism is impartial.

    We now claim that the mechanism is $2105/3147$-optimal on $\calG$. Let $n\in\NN$ and $G\in\calG_n$. If $n\leq 5$, the claim follows directly from \autoref{lem:random-dict}: $1/2+1/n$ is decreasing in $n$ and $1/2+1/5=7/10$, so the random dictatorship mechanism and thus the mixture mechanism is $7/10$-optimal on $\calG_n$. 

    Now assume that $n\geq 6$. If $\Delta\geq 4$, $\alpha_1(\Delta)$-optimality of the permutation mechanism follows from \autoref{thm:lb-permutation} and $\alpha_2(\Delta)$-optimality of $\mathsf{PRUG^D}$ from \autoref{lem:prugd} where
    \[
        \alpha_1(\Delta) = \left\{ \begin{array}{ll} 
        \frac{3\Delta+2}{4\Delta+4} & \text{if $\Delta$ is even,} \\[1ex]
        \alpha_1(\Delta-1) & \text{otherwise,}
        \end{array}
        \right. \quad \alpha_2(\Delta) = \frac{1}{2} + \frac{7\Delta-9}{6\Delta(3\Delta-2)}.
    \]
    For $\Delta\geq 4$, $\alpha_2$ is decreasing in $\Delta$ since
    \[
        \alpha_2'(\Delta) = \frac{7\Delta(3\Delta-2)-(7\Delta-9)(6\Delta-2)}{6\Delta^2(3\Delta-2)^2} = -\frac{7\Delta^2 - 18\Delta +6}{2\Delta^2(3\Delta-2)^2} < 0.
    \]
    Moreover, for odd values of $\Delta$, $\alpha_1(\Delta)=\alpha_1(\Delta-1)= 3/4-1/(4\Delta)$. This means that any convex combination of $\alpha_1(\Delta)$ and $\alpha_2(\Delta)$ attains its minimum at an odd value of $\Delta$, and thus
    \begin{align*}
        \frac{\EE_{v\sim \mathsf{Mix}(G)}[\delta^-(v)]}{\Delta} & \geq \min_{\Delta\geq 4} \left\{\frac{825}{1049}\frac{\EE_{v\sim \mathsf{Perm}(G)}[\delta^-(v)]}{\Delta} + \frac{224}{1049} \frac{\EE_{v\sim \mathsf{PRUG}^D}[\delta^-(v)] }{\Delta} \right\}\\
        & \geq \min_{\Delta\geq 4} \left\{\frac{825}{1049} \alpha_1(\Delta) + \frac{224}{1049} \alpha_2(\Delta)\right\}\\
        & \geq \min_{\Delta\geq 4, \Delta \text{ odd}} \left\{ \frac{825}{1049} \left(\frac{3}{4}-\frac{1}{4\Delta}\right) + \frac{224}{1049} \left( \frac{1}{2} + \frac{7\Delta-9}{4\Delta(3\Delta-2)} \right) \right\}\\
        & \geq \frac{2923}{4196} - \max_{\Delta\geq 5} \frac{907\Delta+366}{4196\Delta(3\Delta-2)}.
    \end{align*}
    Defining $g(\Delta)=(907\Delta+366)/(4196\Delta(3\Delta-2))$,
    \[
        g'(\Delta) = \frac{907\Delta(3\Delta-2)-(907\Delta+366)(6\Delta-2)}{4196\Delta^2(3\Delta-2)^2} = -\frac{2721\Delta^2+2196\Delta-732}{4196\Delta^2(3\Delta-2)^2} < 0,
    \]
    so the maximum is attained for $\Delta=5$. Thus, if $\Delta\geq 4$,
    \[
        \frac{\EE_{v\sim \mathsf{Mix}(G)}[\delta^-(v)]}{\Delta} \geq \frac{2923}{4196} - \frac{907\cdot 5+366}{4196\cdot 5\cdot (3\cdot 5 -2)} = \frac{2923}{4196} - \frac{377}{20980} = \frac{7119}{10490} > \frac{2105}{3147}.
    \]

    If $\Delta=2$, we know from \autoref{thm:lb-permutation} that the permutation mechanism is $2/3$-optimal, and from \autoref{lem:prugd} that $\mathsf{PRUG^D}$ is $65/96$-optimal. This implies
    \[
        \frac{\EE_{v\sim \mathsf{Mix}(G)}[\delta^-(v)]}{\Delta} \geq \frac{825}{1049}\cdot \frac{2}{3} + \frac{224}{1049} \cdot \frac{65}{96} = \frac{550}{1049}+\frac{455}{3147} = \frac{2105}{3147}.
    \]

    Finally assume that $\Delta=3$. If $|\{v\in V: \delta^-(v)\geq 2\}|=1$, \autoref{thm:lb-permutation} and \autoref{lem:prugd} guarantee $2/3$-optimality of the permutation mechanism and $13/18$-optimality of $\mathsf{PRUG^D}$, so
    \[
        \frac{\EE_{v\sim \mathsf{Mix}(G)}[\delta^-(v)]}{\Delta} \geq \frac{825}{1049}\cdot \frac{2}{3} + \frac{224}{1049} \cdot \frac{13}{18} = \frac{550}{1049}+\frac{1456}{9441} = \frac{6406}{9441} > \frac{2105}{3147}.
    \]
    If $|\{v\in V: \delta^-(v)\geq 2\}|\geq 2$, \autoref{lem:tight-instances-permutation} guarantees $31/45$-optimality of the permutation mechanism and \autoref{lem:prugd} guarantees $25/42$-optimality of $\mathsf{PRUG^D}$, so
    \[
        \frac{\EE_{v\sim \mathsf{Mix}(G)}[\delta^-(v)]}{\Delta} \geq \frac{825}{1049}\cdot \frac{31}{45} + \frac{224}{1049} \cdot \frac{25}{42} = \frac{1705}{3147}+\frac{400}{3147} = \frac{2105}{3147}.
    \]
    This shows that the mixture mechanism is $2105/3147$-optimal on $\calG$ and completes the proof.
\end{proof}

\definecolor{color1}{HTML}{1b9e77}
\definecolor{color2}{HTML}{d95f02}
\definecolor{color3}{HTML}{7570b3}
\definecolor{color4}{HTML}{e7298a}
\pgfdeclareplotmark{halfcircle}{%
\pgfpathcircle{\pgfpoint{0pt}{0pt}}{\pgfplotmarksize}
\pgfusepathqstroke
\pgfpathmoveto{\pgfpoint{\pgfplotmarksize}{0pt}}
\pgfpatharc{0}{180}{\pgfplotmarksize}
\pgfpathclose
\pgfusepathqfill
}
\pgfdeclareplotmark{thirdcircle}{%
\pgfpathcircle{\pgfpoint{0pt}{0pt}}{\pgfplotmarksize}
\pgfusepathqstroke
\pgfpathmoveto{\pgfpointorigin}
\pgfpathlineto{\pgfpoint{\pgfplotmarksize}{0pt}}
\pgfpatharc{0}{120}{\pgfplotmarksize}
\pgfpathclose
\pgfusepathqfill
}
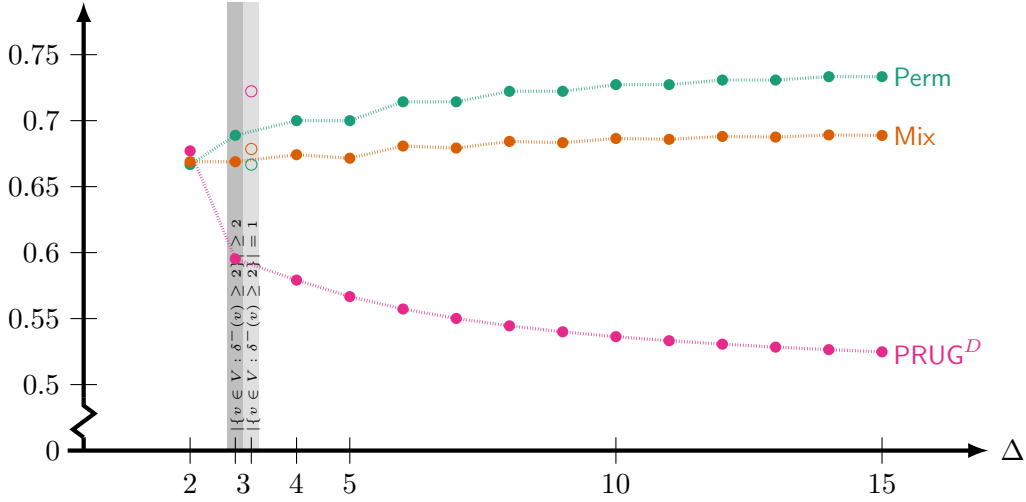
\begin{figure}[t]
    \centering
\begin{tikzpicture}[xscale=0.7,yscale=0.7]

\fill[lightgray] (2.7,0) rectangle (3,8.5);
\fill[lightgray,fill opacity=0.5] (3,0) rectangle (3.3,8.5);
\node[anchor=west,rotate=90,font=\fontsize{4.5}{4.5}\selectfont] (ass1) at (2.85,0.15) {$|\{v \in V : \delta^-(v) \geq 2\}| \geq 2$};
\node[anchor=west,rotate=90,font=\fontsize{4.5}{4.5}\selectfont] (ass2) at (3.15,0.15) {$|\{v \in V : \delta^-(v) \geq 2\}| = 1$};
\draw[ultra thick,-latex] (-.3,0) -- (17,0) node[right] {$\Delta$};
\discontarrow(0,0)(0,0.25)(0,1)(0,8.5);
\foreach \n in {2,4,5,10,15}{
    \draw (\n,0.25) -- (\n,-0.25) node[below] {$\n$};
}
\draw (2.85,0.25) -- (2.85,-0.25);
\draw (3.15,0.25) -- (3.15,-0.25);
\path (3,0.25) -- (3,-0.25) node[below] {$3$};
\draw (0.25,0) -- (-0.25,0) node[left] {$0$};
\draw (0.25,1.25) -- (-0.25,1.25) node[left] {$0.5$};
\draw (0.25,2.5) -- (-0.25,2.5) node[left] {$0.55$};
\draw (0.25,3.75) -- (-0.25,3.75) node[left] {$0.6$};
\draw (0.25,5) -- (-0.25,5) node[left] {$0.65$};
\draw (0.25,6.25) -- (-0.25,6.25) node[left] {$0.7$};
\draw (0.25,7.5) -- (-0.25,7.5) node[left] {$0.75$};

\draw[very thick,color1,dash pattern=on 0.25pt off 0.75pt] plot[mark=*, mark color=none,mark options={fill=color1,draw opacity=0},mark size=3pt] coordinates { (2, 5.417) (2.85, 5.972) (4, 6.25) (5, 6.25) (6, 6.607) (7, 6.607) (8, 6.806) (9, 6.806) (10, 6.932) (11, 6.932) (12, 7.019) (13, 7.019) (14, 7.083) (15, 7.083)   } node[right] {\small $\mathsf{Perm}$};
\draw[very thick,color4,dash pattern=on 0.25pt off 0.75pt, dash phase=0.25pt] plot[mark=*, mark color=none, mark options={color=color4,draw opacity=0},mark size=3pt] coordinates {(2, 5.677) (2.85, 3.631) (4, 3.229) (5, 2.917) (6, 2.682) (7, 2.503) (8, 2.363) (9, 2.25) (10, 2.158) (11, 2.081) (12, 2.016) (13, 1.96) (14, 1.912) (15, 1.87)   } node[right] {\small $\mathsf{PRUG}^D$};
\draw[very thick,color2,dash pattern=on 0.25pt off 0.75pt, dash phase=0.5pt] plot[mark=*, mark color=none,mark options= {color=color2,draw opacity=0},mark size=3pt] coordinates { (2, 5.472) (2.85, 5.472) (4, 5.605) (5, 5.538) (6, 5.769) (7, 5.731) (8, 5.857) (9, 5.833) (10, 5.912) (11, 5.896) (12, 5.951) (13, 5.939) (14, 5.979) (15, 5.97)  } node[right] {\small $\mathsf{Mix}$};

\path plot[mark=o, mark options = {color=color1,draw opacity=1
}, mark size =3pt] coordinates { (3.15, 5.417)};
\path plot[mark=o, mark options = {color=color4,draw opacity=1
}, mark size =3pt] coordinates {(3.15, 6.806) };
\path plot[mark=o, mark options = {color=color2,draw opacity=1
}, mark size =3pt] coordinates {(3.15, 5.713) };
\end{tikzpicture}
\caption{Values $\alpha(\Delta)$, for $\Delta\in \{2,\ldots,15\}$, such that the permutation mechanism, $\mathsf{PRUG}^D$, and the mixture mechanism are $\alpha(\Delta)$-optimal on $\{G=(V,E)\in \calG: |V|\geq 6,~ \Delta(G)=\Delta\}$.
The case where $\Delta=3$ requires a case distinction based on the number of vertices with indegree greater than~$1$.}
\label{fig:plot-alpha-delta}
\end{figure}
\autoref{fig:plot-alpha-delta} illustrates the performance of the permutation mechanism, $\mathsf{PRUG^D}$, and the mixture mechanism on graphs with at least six vertices and maximum indegree $\Delta$, for $\Delta\in\{2,\ldots,15\}$.

\section{Upper Bounds}
\label{sec:upper-bound}

We conclude by giving new upper bounds on the performance guarantee of any impartial mechanism. They improve on the bounds of \citet{fischer2015optimal} for all $n\geq 6$ and have a global minimum of $76/105\approx 0.7238$ at $n=7$. Like the bounds of \citeauthor{fischer2015optimal} they approach $3/4$ as $n\to\infty$ and thus do not rule out the existence of an impartial mechanism that is close to $3/4$-optimal when $n$ is large. Whether such a mechanism exists is an intriguing open question, the resolution of which is likely to require new insights.
\begin{theorem}
\label{thm:ub-plurality}
    Let $n\in \NN, n\geq 6$. Let $f$ be a mechanism that is impartial and $\alpha$-optimal on $\calG_n$.
    Then,
    \[
        \alpha \leq \frac{3n^3-19n^2+30n-4}{4n(n-2)(n-4)}.
    \]
\end{theorem}

\begin{corollary}
    If $f$ is an impartial and $\alpha$-optimal mechanism on $\calG$, then $\alpha\leq 76/105$.
    If $f$ is an impartial and $\alpha$-optimal mechanism on $\calG_n$, then $\alpha \leq 3/4-\Theta(1/n)$. 
\end{corollary}

A useful observation when attempting to prove upper bounds on the performance of impartial mechanisms is that we can restrict attention to mechanisms that are symmetric.
Here, mechanism~$f$ is called \emph{symmetric} if it is invariant with respect to renaming of the vertices, \ie if for every $G=(V,E)\in\calG$, every $v\in V$, and every permutation $\pi=(\pi_1,\ldots, \pi_{n})$ of $V$,
\[
    f_{\pi_v}(G_{\pi}) = f_v(G),
\]
where $G_{\pi}=(V,E_{\pi})$ with $E_{\pi}=\{(\pi_u, \pi_v):(u,v)\in E\}$. 
For a given mechanism $f$, denote by $f_s$ the ``symmetrized'' mechanism obtained by renaming vertices according permutation $\pi$ chosen uniformly at random, applying~$f$ to the resulting graph, and renaming the result back by the inverse of $\pi$, such that for all $n\in \NN$ , $G\in\calG_n$, and $v\in\{1,\ldots,n\}$,
\[
    (f_s(G))_v = \frac{1}{n!} \sum_{\pi\in\calP_n}f_{\pi_v}(G_{\pi}),
\]
where $\calP_n$ is the set of all permutations of $V$. It is not difficult to see that this operation does not affect impartiality and approximate optimality and produces a symmetric mechanism.
\begin{lemma}[\citealp{holzman2013impartial}]
\label{lem:symmetry}
Let $f$ be a selection mechanism that is impartial and $\alpha$-optimal on $\calG' \subseteq \calG$. Then $f_s$ is impartial, $\alpha$-optimal, and symmetric on $\calG'$.
\end{lemma}

We are now ready to prove \autoref{thm:ub-plurality}.
\begin{proof}[Proof of \autoref{thm:ub-plurality}]
    Let $n\in \NN$, $n\geq 6$. Let $f$ be an impartial $\alpha$-optimal mechanism on $\calG_n$, which by \autoref{lem:symmetry} we may assume to be symmetric.
    Let $n'=\lfloor n/2\rfloor-1$, and consider the family of graphs $\{G_i=(V,E_i)\}_{i=0}^{n'}\subset \calG_n$, where $V=[n]$ and for each $i\in\{0,\ldots,n'\}$,
    \[
        E_i = \{(v+1,v): v \in \{1,\ldots,n-1\}\setminus \{2,i+2\} \} \cup \{(1,2), (3,1+\chi(i=0)), (i+3,2)\},
    \]
    where $\chi(\cdot)$ denotes the indicator function.
    In other words, in $G_i$, vertices $1$ and $2$ form a $2$-cycle, vertices $i+2$ down to $3$ a path directed at vertex~$1$, and vertices $n$ to $i+3$ a path directed at vertex~$2$. Let $C_n=(V,E(C_n)),C_{2,n}=(V,E(C_{2,n})) \in \calG_n$ with
    \begin{align*}
        E(C_n) & = \{(v+1,v): v\in \{1,\ldots,n-1\}\} \cup \{(1,n)\},\\
        E(C_{2,n}) & = \{(v+1,v): v\in \{3,\ldots,n-1\}\} \cup \{(1,2),(2,1),(3,n)\} \}.
    \end{align*}
    Graphs $G_0$, $G_1$, $G_2$, $C_n$, and $C_{2,n}$ for $n=7$ are shown in \autoref{fig:thm-ub-plurality}. For $v\in V$, let $p_v=f_v(G_0)$ be the probability with which~$f$ selects vertex $v$ in graph $G_0$.

    Then
    \begin{equation}  \label{eq:ub_p1}
        p_1 = f_1(C_n) = 1/n,
    \end{equation}
    where the first equality holds by impartiality, since $E_0 \setminus (\{1\}\times V) = E(C_n) \setminus (\{1\}\times V)$, 
    and the second equality by symmetry. Similarly,
    \begin{equation}  \label{eq:ub_p3}
        p_3 = f_3(C_{2,n}) \leq 1/(n-2),
    \end{equation}
    where the equality holds by impartiality, since $E_0 \setminus (\{3\}\times V) = E(C_{2,n}) \setminus (\{3\}\times V)$,
    and the inequality by symmetry. We further claim that for every $i\in\{0,\ldots n'-1\}$,
    \begin{equation}\label{impartiality-12}
        f_2(G_i) = f_1(G_{i+1}).
    \end{equation}
    To see this, observe that for $i\in\{0,\ldots n'-1\}$,
    \begin{align*}
        E_i \setminus (\{2\}\times V) = & \ \{(v+1,v): v \in \{1,\ldots,n-1\}\setminus \{1,2,i+2\} \}\\ & \ \cup \{(1,2), (3,1+\chi(i=0)), (i+3,2)\},\\
        E_{i+1} \setminus (\{1\}\times V) = & \ \{(v+1,v): v \in \{1,\ldots,n-1\}\setminus \{2,i+3\} \} \cup \{(3,1), (i+4,2)\}.
    \end{align*}
    Both $(V,E_i \setminus (\{2\}\times V))$ and $(V,E_{i+1} \setminus (\{1\}\times V))$ consist of two paths of lengths $i+1$ and $n-i-2$ directed at the same vertex, vertex $2$ in the former graph and vertex $1$ in the latter.
    Denoting by $\pi$ the permutation with
    \[
        \pi_v = \left\{ \begin{array}{ll}
            3 & \text{if $v=1$},\\
            1 & \text{if $v=2$}, \\
            v+1 & \text{if $v\in\{3,\ldots,i+2\}$},\\
            2 & \text{if $v=i+3$},\\
            v & \text{otherwise,}
             \end{array}\right.
    \]
    we have that $(E_i \setminus (\{2\}\times V)))_{\pi} = E_{i+1} \setminus (\{1\}\times V)$. 
    Since $\pi_{2}=1$, and by impartiality, $f_2(G_i)=f_1(G_{i+1})$.
    This implies in particular that $p_2=f_1(G_1)$.
    Denote $x_i=f_2(G_i)$ for $i\in\{1,\ldots,n'\}$, and observe that $x_i=f_1(G_{i+1})$ for $i\in \{1,\ldots,n'-1\}$.

    \begin{figure}[t]
    \centering
    \begin{tikzpicture}  
    \Vertex[x=2.5, y=3.0, Math, shape=circle, color=black, size=.05, label=p_1, fontscale=1, position=left, distance=-.08cm]{A}
    \Vertex[x=2.029, y=3.977, Math, shape=circle, color=black, size=.05, label=p_1, fontscale=1, position={below left}, distance=-.13cm]{B} 
    \Vertex[x=0.972, y=4.219, Math, shape=circle, color=black, size=.05, label=p_1, fontscale=1, position=below, distance=-.08cm]{C}
    \Vertex[x=0.124, y=3.542, Math, shape=circle, color=black, size=.05, label=p_1, fontscale=1, position=right, distance=-.08cm]{D}
    \Vertex[x=0.124, y=2.458, Math, shape=circle, color=black, size=.05, label=p_1, fontscale=1, position=right, distance=-.08cm]{E}
    \Vertex[x=0.972, y=1.781, Math, shape=circle, color=black, size=.05, label=p_1, fontscale=1, position=above, distance=-.08cm]{F}
    \Vertex[x=2.029, y=2.023, Math, shape=circle, color=black, size=.05, label=p_1, fontscale=1, position={above left}, distance=-.13cm]{G}
        
    \Edge[Direct, color=black, lw=1pt, bend=-15](A)(B)
    \Edge[Direct, color=black, lw=1pt, bend=-15](B)(C)
    \Edge[Direct, color=black, lw=1pt, bend=-15](C)(D)
    \Edge[Direct, color=black, lw=1pt, bend=-15](D)(E)
    \Edge[Direct, color=black, lw=1pt, bend=-15](E)(F)
    \Edge[Direct, color=black, lw=1pt, bend=-15](F)(G)
    \Edge[Direct, color=black, lw=1pt, bend=-15](G)(A)

    \Text[y=4.3]{$C_n$}
    \Text[x=1.25,y=1.2]{$p_1\leq \frac{1}{7}$}

    \Vertex[x=6, y=3.5, Math, shape=circle, color=black, size=.05]{A}
    \Vertex[x=6, y=2.5, Math, shape=circle, color=black, size=.05]{B} 
    \Vertex[x=5.5, y=3.0, Math, shape=circle, color=black, size=.05, label=p_3, fontscale=1, position=left, distance=-.08cm]{C}
    \Vertex[x=4.809, y=3.951, Math, shape=circle, color=black, size=.05, label=p_3, fontscale=1, position=below, distance=-.08cm]{D}
    \Vertex[x=3.691, y=3.588, Math, shape=circle, color=black, size=.05, label=p_3, fontscale=1, position=below right, distance=-.13cm]{E}
    \Vertex[x=3.691, y=2.412, Math, shape=circle, color=black, size=.05, label=p_3, fontscale=1, position=above right, distance=-.13cm]{F}
    \Vertex[x=4.809, y=2.049, Math, shape=circle, color=black, size=.05, label=p_3, fontscale=1, position=above, distance=-.08cm]{G}
        
    \Edge[Direct, color=black, lw=1pt, bend=-20](A)(B)
    \Edge[Direct, color=black, lw=1pt, bend=-20](B)(A)
    \Edge[Direct, color=black, lw=1pt, bend=-20](C)(D)
    \Edge[Direct, color=black, lw=1pt, bend=-20](D)(E)
    \Edge[Direct, color=black, lw=1pt, bend=-20](E)(F)
    \Edge[Direct, color=black, lw=1pt, bend=-20](F)(G)
    \Edge[Direct, color=black, lw=1pt, bend=-20](G)(C)

    \Text[x=3.5, y=4.3]{$C_{2,n}$}
    \Text[x=4.75,y=1.2]{$p_3\leq \frac{1}{5}$}
    
    \Vertex[x=9.5, y=3, Math, shape=circle, color=black, size=.05]{A}
    \Vertex[x=9.5, y=4, Math, shape=circle, color=black, size=.05]{B} 
    \Vertex[x=9.0, y=3.0, Math, shape=circle, color=black, size=.05, label=x_1, fontscale=1, position=left, distance=-.08cm]{C}
    \Vertex[x=8.309, y=3.951, Math, shape=circle, color=black, size=.05]{D}
    \Vertex[x=7.191, y=3.588, Math, shape=circle, color=black, size=.05]{E}
    \Vertex[x=7.191, y=2.412, Math, shape=circle, color=black, size=.05]{F}
    \Vertex[x=8.309, y=2.049, Math, shape=circle, color=black, size=.05]{G}
        
    \Edge[Direct, color=black, lw=1pt](A)(B)
    \Edge[Direct, color=black, lw=1pt](B)(C)
    \Edge[Direct, color=black, lw=1pt, bend=-20](C)(D)
    \Edge[Direct, color=black, lw=1pt, bend=-20](D)(E)
    \Edge[Direct, color=black, lw=1pt, bend=-20](E)(F)
    \Edge[Direct, color=black, lw=1pt, bend=-20](F)(G)
    \Edge[Direct, color=black, lw=1pt, bend=-20](G)(C)

    \Text[x=7, y=4.3]{$G'_1$}
    \Text[x=8.25,y=1.2]{$\alpha\leq \frac{1}{2}(x_1+1)$}

    \Vertex[x=13, y=2, Math, shape=circle, color=black, size=.05]{A}
    \Vertex[x=13, y=3, Math, shape=circle, color=black, size=.05]{B} 
    \Vertex[x=13, y=4, Math, shape=circle, color=black, size=.05]{C}
    \Vertex[x=12.5, y=3, Math, shape=circle, color=black, size=.05, label=x_2, fontscale=1, position=left, distance=-.08cm]{D}
    \Vertex[x=11.5, y=4, Math, shape=circle, color=black, size=.05]{E}
    \Vertex[x=10.5, y=3, Math, shape=circle, color=black, size=.05]{F}
    \Vertex[x=11.5, y=2, Math, shape=circle, color=black, size=.05]{G}
        
    \Edge[Direct, color=black, lw=1pt](A)(B)
    \Edge[Direct, color=black, lw=1pt](B)(C)
    \Edge[Direct, color=black, lw=1pt](C)(D)
    \Edge[Direct, color=black, lw=1pt, bend=-30](D)(E)
    \Edge[Direct, color=black, lw=1pt, bend=-30](E)(F)
    \Edge[Direct, color=black, lw=1pt, bend=-30](F)(G)
    \Edge[Direct, color=black, lw=1pt, bend=-30](G)(D)

    \Text[x=10.5, y=4.3]{$G'_2$}
    \Text[x=11.75,y=1.2]{$\alpha\leq \frac{1}{2}(x_2+1)$}

    \Vertex[y=8.5, Math, shape=circle, color=black, size=.05, label=p_7, fontscale=1, position=left, distance=-.08cm]{A}
    \Vertex[y=7.5, Math, shape=circle, color=black, size=.05, label=p_6, fontscale=1, position=left, distance=-.08cm]{B} 
    \Vertex[y=6.5, Math, shape=circle, color=black, size=.05, label=p_5, fontscale=1, position=below, distance=-.08cm]{C}
    \Vertex[x=1, y=6.5, Math, shape=circle, color=black, size=.05, label=p_4, fontscale=1, position=below, distance=-.08cm]{D}
    \Vertex[x=2, y=6.5, Math, shape=circle, color=black, size=.05, label=p_3, fontscale=1, position=below, distance=-.08cm]{E}
    \Vertex[x=3, y=6.5, Math, shape=circle, color=black, size=.05, label=p_2, fontscale=1, position=below, distance=-.08cm]{F}
    \Vertex[x=3, y=8.5, Math, shape=circle, color=black, size=.05, label=p_1, fontscale=1, position=above, distance=-.08cm]{G}
        
    \Edge[Direct, color=black, lw=1pt](A)(B)
    \Edge[Direct, color=black, lw=1pt](B)(C)
    \Edge[Direct, color=black, lw=1pt](C)(D)
    \Edge[Direct, color=black, lw=1pt](D)(E)
    \Edge[Direct, color=black, lw=1pt](E)(F)
    \Edge[Direct, color=black, lw=1pt, bend=-20](F)(G)
    \Edge[Direct, color=black, lw=1pt, bend=-20](G)(F)

    \Text[x=0.3, y=8.8]{$G_0$}
    \Text[x=1.5,y=5.7]{$\sum_{v=1}^{7} p_v = 1$}

    \Vertex[x=7, y=8.5, Math, shape=circle, color=black, size=.05, label=p_7, fontscale=1, position=above, distance=-.08cm]{A}
    \Vertex[x=5, y=7.5, Math, shape=circle, color=black, size=.05]{B} 
    \Vertex[x=5, y=6.5, Math, shape=circle, color=black, size=.05]{C}
    \Vertex[x=6, y=6.5, Math, shape=circle, color=black, size=.05]{D}
    \Vertex[x=7, y=6.5, Math, shape=circle, color=black, size=.05, label=p_4, fontscale=1, position=below, distance=-.08cm]{E}
    \Vertex[x=8, y=6.5, Math, shape=circle, color=black, size=.05, label=x_1, fontscale=1, position=below, distance=-.08cm]{F}
    \Vertex[x=8, y=8.5, Math, shape=circle, color=black, size=.05, label=p_2, fontscale=1, position=above, distance=-.08cm]{G}
        
    \Edge[Direct, color=black, lw=1pt](A)(G)
    \Edge[Direct, color=black, lw=1pt](B)(C)
    \Edge[Direct, color=black, lw=1pt](C)(D)
    \Edge[Direct, color=black, lw=1pt](D)(E)
    \Edge[Direct, color=black, lw=1pt](E)(F)
    \Edge[Direct, color=black, lw=1pt, bend=-20](F)(G)
    \Edge[Direct, color=black, lw=1pt, bend=-20](G)(F)

    \Text[x=5.3, y=8.8]{$G_1$}
    \Text[x=6.5,y=5.7]{$p_2+p_4+p_7+x_1 \leq 1$}

    \Vertex[x=11, y=8.5, Math, shape=circle, color=black, size=.05]{A}
    \Vertex[x=12, y=8.5, Math, shape=circle, color=black, size=.05, label=p_6, fontscale=1, position=above, distance=-.08cm]{B} 
    \Vertex[x=10, y=6.5, Math, shape=circle, color=black, size=.05]{C}
    \Vertex[x=11, y=6.5, Math, shape=circle, color=black, size=.05]{D}
    \Vertex[x=12, y=6.5, Math, shape=circle, color=black, size=.05, label=p_5, fontscale=1, position=below, distance=-.08cm]{E}
    \Vertex[x=13, y=6.5, Math, shape=circle, color=black, size=.05, label=x_2, fontscale=1, position=below, distance=-.08cm]{F}
    \Vertex[x=13, y=8.5, Math, shape=circle, color=black, size=.05, label=x_1, fontscale=1, position=above, distance=-.08cm]{G}
        
    \Edge[Direct, color=black, lw=1pt](A)(B)
    \Edge[Direct, color=black, lw=1pt](B)(G)
    \Edge[Direct, color=black, lw=1pt](C)(D)
    \Edge[Direct, color=black, lw=1pt](D)(E)
    \Edge[Direct, color=black, lw=1pt](E)(F)
    \Edge[Direct, color=black, lw=1pt, bend=-20](F)(G)
    \Edge[Direct, color=black, lw=1pt, bend=-20](G)(F)

    \Text[x=10.3, y=8.8]{$G_2$}
    \Text[x=11.5, y=5.7]{$p_5+p_6+x_1+x_2 \leq 1$}
    \end{tikzpicture}
    \caption{Set of graphs that yield $\alpha\leq 76/105$ for any impartial $\alpha$-optimal mechanism on $\calG_7$. Impartiality forces the probability assignment shown in the figure. The equality and inequalities come from the fact that probabilities sum up to 1 in each graph or from imposing $\alpha$-optimality. Combining all the expressions yields the bound stated in \autoref{thm:ub-plurality} for $n=7$.}    \label{fig:thm-ub-plurality}
\end{figure}
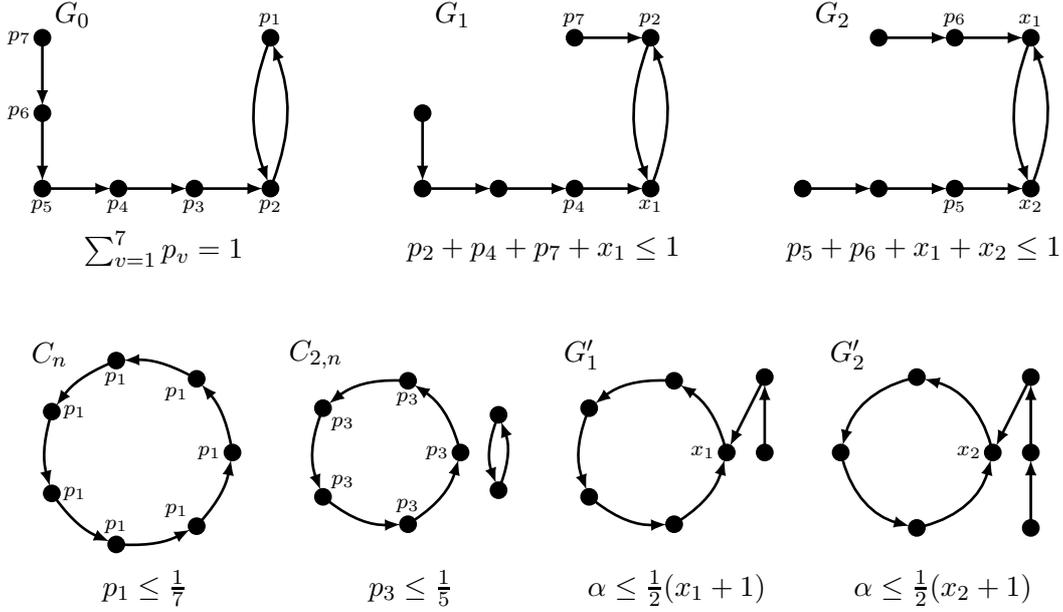

    We finally claim that for every $i\in \{1,\ldots,n'\}$,
    \begin{equation}\label{impartiality-p_i}
        f_3(G_i)=p_{n-i+1} \quad\text{and}\quad f_{i+3}(G_i)) = p_{i+3}.
    \end{equation}
    For the first equality,     
    observe that for $i\in\{1,\ldots,n'\}$,
    \begin{align*}
        E_i \setminus (\{3\} \times V)& =  \{(v+1,v): v \in \{1,\ldots,n-1\}\setminus \{2,i+2\} \} \cup \{(1,2),(i+3,2)\},\\
        E_0 \setminus (\{n-i+1\} \times V) &= \{(v+1,v): v \in \{1,\ldots,n-1\}\setminus \{n-i\} \} \cup \{(1,2)\}.
    \end{align*}
    Denoting by $\pi$ the permutation with
    \[
        \pi_v = \left\{ \begin{array}{ll}
            v & \text{if $v\in \{1,2\}$},\\
            v+n-i-2 & \text{if $v\in \{3,\ldots,i+2\}$}, \\
            v-i &  \text{otherwise,}
             \end{array}\right.
    \]
    we have that $(E_i \setminus (\{3\} \times V))_{\pi} = E_0 \setminus (\{n-i+1\} \times V)$. Since $\pi_3=n-i+1$, and by impartiality, $f_3(G_i)=p_{n-i+1}$.
    For the second equality %
    observe that for $i\in \{1,\ldots,n'\}$,
    \begin{align*}
        E_i \setminus (\{i+3\} \times V)& =  \{(v+1,v): v \in \{1,\ldots,n-1\}\setminus \{2,i+2\} \} \cup \{(1,2), (3,1)\},\\
        E_0 \setminus (\{i+3\} \times V) &= \{(v+1,v): v \in \{1,\ldots,n-1\}\setminus \{i+2\} \} \cup \{(1,2)\}.
    \end{align*}
    Denoting by $\pi$ the permutation with 
    \[
        \pi_v = \left\{ \begin{array}{ll}
            2 & \text{if $v=1$},\\
            1 &  \text{if $v=2$}, \\
            v &  \text{otherwise,}
             \end{array}\right.
    \]
    we have that $(E_i \setminus (\{i+3\} \times V))_{\pi} = E_0 \setminus (\{i+3\} \times V)$. 
    Since $\pi_{i+3}=\pi_{i+3}$, and by impartiality, $f_{i+3}(G_i)=p_{i+3}$.

    By combining what we have shown above with the fact that selection probabilities in each graph sum to~$1$, we obtain
    \begin{align}
        p_2+\sum_{v=4}^n p_v & \geq \frac{n^2-4n+2}{n(n-2)},\label{eq:G_0}\\
        p_2+p_4+p_n+x_1 & \leq 1,\label{eq:G_1}\\
        p_{n-i+1}+p_{i+3}+x_{i-1}+x_i & \leq 1\quad \text{ for all } i\in \left\{2,\ldots n' \right\}. \label{eq:G_i}
    \end{align}
    Indeed, \eqref{eq:G_0} holds because probabilities for $G_0$ sum to $1$, together with \eqref{eq:ub_p1} and \eqref{eq:ub_p3}; \eqref{eq:G_1} holds because probabilities for $G_1$ sum to at most $1$, together with \eqref{impartiality-12} and \eqref{impartiality-p_i}; \eqref{eq:G_i} holds because probabilities for $G_i$ with $i\in\{2,\ldots,n'\}$ sum to at most $1$, together with \eqref{impartiality-12} and \eqref{impartiality-p_i}.
    
    If $n$ is odd, then $n'=(n-3)/2$ and $n'+4=(n+5)/2=n-n'+1$.
    Therefore, $p_v$ for each $v\in \{2\} \cup \{4,\ldots,n\}$ appears exactly once in \eqref{eq:G_1} and \eqref{eq:G_i}. Adding up \eqref{eq:G_1} and \eqref{eq:G_i} for $i\in\{2,\ldots,n'\}$ yields
    \begin{equation*}
        p_2 + \sum_{v=4}^{n'}p_v + 2\sum_{i=1}^{n'-1}x_i + x_{n'} \leq n',
    \end{equation*}
    which together with \eqref{eq:G_0} implies
    \begin{equation}\label{eq:ub-x-n-odd}
        2\sum_{i=1}^{n'-1}x_i + x_{n'}  \leq \frac{n-3}{2} - \frac{n^2-4n+2}{n(n-2)}.
    \end{equation}
    
    If $n$ is even, then $n' = n/2-1$ and $n'+3 = n/2+2 = n-n'+1$, $p_{n/2+2}$ appears twice in \eqref{eq:G_i} for $i=n'$, and $p_v$ for each $v\in \{2\} \cup \{4,\ldots,n\}\setminus \{n/2+2\}$ appears exactly once in \eqref{eq:G_1} and \eqref{eq:G_i}.
    Adding \eqref{eq:G_1} multiplied by $2$, \eqref{eq:G_i} multiplied by $2$ for $i\in \{2,\ldots,n'-1\}$, and \eqref{eq:G_i} for $i=n'$ yields
    \begin{equation}
        2p_2 + 2\sum_{v=4}^{n'}p_v + 4\sum_{i=1}^{n'-2}x_i + 3x_{n'-1} + x_{n'} \leq 2n'-1,
    \end{equation}
    which together with \eqref{eq:G_0} implies
    \begin{equation}\label{eq:ub-x-n-even}
        4\sum_{i=1}^{n'-2}x_i + 3x_{n'-1} + x_{n'}  \leq n-3 - \frac{2(n^2-4n+2)}{n(n-2)}.
    \end{equation}
    We now claim that
    \begin{equation}
        \alpha \leq \frac{1}{2}\left(\min_{i\in \{1,\ldots,n'\}}x_i+1\right).\label{bound-alpha-x}
    \end{equation}
    To see this, let $i^*\in\arg\min_{\{1,\ldots,n'\}}x_i$ and consider the graph $G'_{i^*}=(V,E'_{i^*})\in \calG_n$ with $E'_{i^*}= (E_{i^*}\setminus \{(2,1)\}) \cup \{(2,n)\}$.
    By impartiality, $f_2(G'_{i^*})=f_2(G_{i^*})=x_{i^*}$.
    Moreover, $\Delta(G'_{i^*})=2$ and $\arg\max_{v\in V} \delta^-(v,G'_{i^*}) = \{2\}$.
    Thus
    \[
        \EE_{v\sim f(G'_{i^*})} [\delta^-(v, G'_{i^*})] \leq 2x_{i^*} + 1\cdot (1-x_{i^*}) = x_{i^*}+1,
    \]
    and $\alpha$-optimality of $f$ implies \eqref{bound-alpha-x}. Graphs $G'_1$ and $G'_2$ for $n=7$ are shown in \autoref{fig:thm-ub-plurality}.

    Define $\mathit{avg}_x:\RR^{n'}_{>0}\to\RR_{>0}$ such that $\mathit{avg}_x(w)$ is the weighted average of $x$ with weights $w$, \ie for all $x,w\in \RR^n_{>0}$,
    \[
        \mathit{avg}_x(w) = \frac{1}{\sum_{i=1}^{n'}w_i} \sum_{i=1}^{n'} w_ix_i.
    \]
    Then, for every $w\in\RR^n_{>0}$, $\min_{i\in \{1,\ldots,n'\}}x_i \leq \mathit{avg}_x(w)$.
    If $n$ is odd, then by \eqref{eq:ub-x-n-odd},
    \[
        \mathit{avg}_x(2,\ldots,2,1) \leq \frac{1}{2n'-1}\left[\frac{n-3}{2} - \frac{n^2-4n+2}{n(n-2)}\right] = \frac{1}{2(n-4)}\left[n-3 - \frac{2(n^2-4n+2)}{n(n-2)} \right].
    \]
    If $n$ is even, then by \eqref{eq:ub-x-n-even},
     \[
        \mathit{avg}_x(4,\ldots,4,3,1) \leq \frac{1}{4n'-4}\left[ n-3 - \frac{2(n^2-4n+2)}{n(n-2)} \right] = \frac{1}{2(n-4)}\left[ n-3 - \frac{2(n^2-4n+2)}{n(n-2)} \right].
    \]
    In both cases, we conclude from \eqref{bound-alpha-x} that
    \[
        \alpha \leq \frac{1}{2}\left(1+\frac{1}{2(n-4)}\left[ n-3 - \frac{2(n^2-4n+2)}{n(n-2)} \right]\right) = \frac{3n^3-19n^2+30n-4}{4n(n-2)(n-4)}.
        \tag*{\raisebox{-.5\baselineskip}{\qedhere}}
    \]
\end{proof}

\appendix

\newpage

\section{Proof of Theorem \ref{thm:ub-permutation}}
\label{app:thm-ub-permutation}

  Let $\Delta\in\NN$ and $\varepsilon>0$. The result holds trivially if $\Delta=1$, so we assume $\Delta\geq 2$ in the following. Our goal will be to construct a graph $G=(V,E)$ such that the permutation mechanism is \emph{not} $(\alpha(\Delta)+\varepsilon)$-optimal on $G$.
    
  Let $n'\in\NN$, to be fixed later,
    and $n=\Delta+1+n'(\lfloor \Delta/2\rfloor +1)$.
    Fix $V=\{1,\ldots,n\}$ and let $E$ be the set
    \begin{align*}
        & \{(v,2v+\Delta+n'-2): v\in \{1,\ldots,n'+1\}\} ~\cup~ \{(v,1): v\in \{n'+2,\ldots,n'+\Delta+1\}\} ~ \cup \\
        & \left\{(u,v)\colon v\in \{2,\ldots,n'+1\}, u\in \left\{n'+\Delta+2+(v-2)\left\lfloor \frac{\Delta}{2}\right\rfloor, \ldots, n'+\Delta+1+(v-1)\left\lfloor \frac{\Delta}{2}\right\rfloor\right\} \right\}.
    \end{align*}
    In words, the graph $G$ consists of $n'+1$ disjoint subgraphs; one subgraph contains vertex~$1$ as well as~$\Delta$ vertices with an outgoing edge directed at vertex~$1$; each of the other $n'$ subgraphs contains a vertex $i\in\{2,\dots,n'+1\}$ as well as $\lfloor\Delta/2\rfloor$ vertices with an outgoing edge directed at vertex~$i$. Vertices~$1$ to~$n'+1$ finally have an outgoing edge directed at the vertex with smallest index among their in-neighbors. An example of~$G$ is shown in \autoref{fig:thm-ub-permutation}.
    
    We note that $\delta^-(1)=\Delta,~\delta^-(v)=\lfloor \Delta/2\rfloor$ for every $v\in \{2,\ldots,n'+1\}$, and $\delta^-(v)=1$ for every $v\in \{n'+2,\ldots,2n'+2\}$. Moreover, $\PP\left[\delta^-_{\pi_{<1}}(1) > \lfloor \Delta/2 \rfloor \right] = (\Delta-\lfloor \Delta/2\rfloor)/(\Delta+1)$, $\PP\left[\delta^-_{\pi_{<1}}(1) \leq \lfloor \Delta/2 \rfloor \right] = (\lfloor \Delta/2\rfloor+1)/(\Delta+1)$, and, by \autoref{lem:max-indegree-left}, $\PP\left[ v^P=1 ~|~ \delta^-_{\pi_{<1}}(1) > \lfloor \Delta/2 \rfloor \right]=1$. The permutation mechanism thus chooses vertex $1$ with probability
    \[
        \PP\left[v^P=1\right] = \frac{\Delta-\lfloor \Delta/2\rfloor}{\Delta+1} + \frac{\lfloor \Delta/2\rfloor+1}{\Delta+1} \PP\left[v^P=1 ~\Bigg|~ \delta^-_{\pi_{<1}}(1) \leq \left\lfloor\frac{\Delta}{2}\right\rfloor \right].
    \]
    We claim that 
    \begin{equation}
        \PP\left[v^P=1 ~\Bigg|~ \delta^-_{\pi_{<1}}(1) \leq \left\lfloor\frac{\Delta}{2}\right\rfloor \right] \leq \left( \frac{2\lfloor \Delta/2 \rfloor +1}{2\lfloor \Delta/2 \rfloor +2}\right)^{n'}\hspace{-1.5ex}.\label{claim:prob-selecting-top}
    \end{equation}
    We will first show that this claim implies the theorem and then prove the claim.

    If~\eqref{claim:prob-selecting-top} holds, then
    \begin{align*}
        \EE[X] & \leq \Delta\PP\left[v^P=1\right] + \left\lfloor \frac{\Delta}{2} \right\rfloor (1-\PP\left[v^P=1\right])\\
        & \leq \left\lfloor \frac{\Delta}{2} \right\rfloor + \left(\Delta - \left\lfloor \frac{\Delta}{2} \right\rfloor \right) \frac{\Delta-\lfloor \Delta/2\rfloor}{\Delta+1} + \left(\Delta - \left\lfloor \frac{\Delta}{2} \right\rfloor \right)\frac{\lfloor \Delta/2\rfloor+1}{\Delta+1} \left( \frac{2\lfloor \Delta/2 \rfloor +1}{2\lfloor \Delta/2 \rfloor +2}\right)^{n'}\hspace{-1.5ex}.
    \end{align*}
    Denote the last term of the sum by $h(n',\Delta)$. 
    Then, if $\Delta$ is even,
    \[
        \frac{\EE[X]}{\Delta} \leq \frac{1}{2} + \frac{\Delta}{4(\Delta+1)} + \frac{1}{\Delta} h(n',\Delta) = \frac{3\Delta+2}{4\Delta+4} + \frac{1}{\Delta} h(n',\Delta) = \alpha(\Delta) + \frac{1}{\Delta} h(n',\Delta).
    \]
    If $\Delta$ is odd,
    \[
        \frac{\EE[X]}{\Delta} \leq \frac{\Delta-1}{2\Delta} + \frac{\Delta+1}{4\Delta} + \frac{1}{\Delta} h(n',\Delta)
        = \frac{3\Delta-1}{4\Delta} + \frac{1}{\Delta} h(n',\Delta) = \alpha(\Delta) + \frac{1}{\Delta} h(n',\Delta).
    \]
    In both cases, the permutation mechanism is \emph{not} $(\alpha(\Delta)+\varepsilon)$-optimal on $G$ if $1/\Delta g(n',\Delta) <\varepsilon$, which happens if and only if
    \[
      \left(\Delta - \left\lfloor \frac{\Delta}{2} \right\rfloor \right)\frac{\lfloor \Delta/2\rfloor+1}{\Delta+1} \left( \frac{2\lfloor \Delta/2 \rfloor +1}{2\lfloor \Delta/2 \rfloor +2}\right)^{n'} <\varepsilon,
    \]
    or, equivalently,
    \[
      n' > \frac{\log\bigl((\Delta-\lfloor \Delta/2 \rfloor)(\lfloor \Delta/2\rfloor +1)\bigr) - \log\bigl((\Delta+1)\varepsilon\bigr)}{\log(2\lfloor \Delta/2 \rfloor +2)-\log(2\lfloor \Delta/2 \rfloor +1)}.
    \]
    The theorem thus follows by choosing $n'$ large enough to satisfy this inequality.

    We now prove~\eqref{claim:prob-selecting-top}.
    Whenever $\delta^-_{\pi_{<1}}(1)\leq \lfloor \Delta/2 \rfloor$ and there exists a vertex $v$ with $1\in \pi_{<v}$ and $\delta^-_{\pi_{<v}}(v)= \lfloor \Delta/2 \rfloor$, the permutation mechanism does not select vertex~$1$. Thus
    \begin{align*}
        \PP\left[v^P=1 ~\Bigg|~ \delta^-_{\pi_{<1}}(1) \leq \left\lfloor\frac{\Delta}{2}\right\rfloor \right] & \leq \PP\left[ \bigcap_{v=2}^{n'+1} \left[ \delta^-_{\pi_{<v}}(v) < \left\lfloor \frac{\Delta}{2}\right\rfloor ~\vee~ v\in \pi_{<1} \right] ~\Bigg|~ \delta^-_{\pi_{<1}}(1) \leq \left\lfloor\frac{\Delta}{2}\right\rfloor \right]\\
        & = \left(1- \PP\left[ \delta^-_{\pi_{<2}}(2) = \left\lfloor \frac{\Delta}{2}\right\rfloor ~\wedge~ 1\in \pi_{<2} ~\Bigg|~ \delta^-_{\pi_{<1}}(1) \leq \left\lfloor\frac{\Delta}{2}\right\rfloor \right] \right)^{n'}\\
        & = \left(1- \frac{1}{\lfloor \Delta/2 \rfloor +1}\PP\left[ 1\in \pi_{<2} ~\Bigg|~ \delta^-_{\pi_{<1}}(1) \leq \left\lfloor\frac{\Delta}{2}\right\rfloor = \delta^-_{\pi_{<2}}(2) \right] \right)^{n'}\hspace{-1.5ex},
    \end{align*}
    where we have used the fact that the indegrees from the left of vertices $1,2,\ldots,n'+1$ are independent random variables, that the indegrees from the left of vertices $2,\ldots,n'+1$ are, in addition, identically distributed, and that $\PP\left[ \delta^-_{\pi_{<2}}(2) = \lfloor \Delta/2 \rfloor \right] = 1/ \left( \lfloor \Delta/2 \rfloor +1 \right)$.
    We now claim that
    \begin{equation}
        \PP\left[ \delta^-_{\pi_{<1}}(1) \leq \left\lfloor\frac{\Delta}{2}\right\rfloor = \delta^-_{\pi_{<2}}(2) ~\Bigg|~ 1\in \pi_{<2} \right] \geq \PP\left[ \delta^-_{\pi_{<1}}(1) \leq \left\lfloor\frac{\Delta}{2}\right\rfloor = \delta^-_{\pi_{<2}}(2) ~\Bigg|~ 2\in \pi_{<1} \right].\label{claim:cond-prob-first-vertex}
    \end{equation}
    This claim, together with Bayes' Theorem and the fact that $\PP\left[ 1\in \pi_{<2} \right] = \PP\left[ 2\in \pi_{<1} \right] = 1/2$, implies that $\PP\left[ 1\in \pi_{<2} ~|~ \delta^-_{\pi_{<1}}(1) \leq \lfloor\Delta/2\rfloor = \delta^-_{\pi_{<2}}(2) \right] \geq 1/2$,
    and substituting this term in the previous expression yields
    \[
        \PP\left[v^P=1 ~\Bigg|~ \delta^-_{\pi_{<1}}(1) \leq \left\lfloor\frac{\Delta}{2}\right\rfloor \right] \leq \left(1- \frac{1}{2(\lfloor \Delta/2 \rfloor +1)} \right)^{n'} = \left( \frac{2\lfloor \Delta/2 \rfloor +1}{2\lfloor \Delta/2 \rfloor +2}\right)^{n'}\hspace{-1.5ex},
    \]
    which is identical to~\eqref{claim:prob-selecting-top}.
    
    It remains to prove \eqref{claim:cond-prob-first-vertex}. For this, we use a similar approach as in the proof of \autoref{lem:correlation-neighbors}.
    We first note that
    \begin{align*}
        \PP\left[ \delta^-_{\pi_{<1}}(1) \leq \left\lfloor\frac{\Delta}{2}\right\rfloor = \delta^-_{\pi_{<2}}(2) ~\Bigg|~ 1\in \pi_{<2} \right] = \frac{\PP\left[ \delta^-_{\pi_{<1}}(1) \leq \lfloor\Delta/2 \rfloor = \delta^-_{\pi_{<2}}(2) ~\wedge~ 1\in \pi_{<2} \right]}{\PP\left[ 1\in \pi_{<2}\right]}
      \intertext{and}
        \PP\left[ \delta^-_{\pi_{<1}}(1) \leq \left\lfloor\frac{\Delta}{2}\right\rfloor = \delta^-_{\pi_{<2}}(2) ~\Bigg|~ 2\in \pi_{<1} \right] = \frac{\PP\left[ \delta^-_{\pi_{<1}}(1) \leq \lfloor\Delta/2 \rfloor = \delta^-_{\pi_{<2}}(2)  ~\wedge~ 2\in \pi_{<1} \right]}{\PP\left[ 2\in \pi_{<1}\right]}.
    \end{align*}
    Since $\PP\left[ 1\in \pi_{<2}\right] = \PP\left[ 2\in \pi_{<1}\right] = 1/2$, it suffices to show the inequality for the numerators.
    Let
    \begin{align*}
        \calP_{1<2}(n) & := \left\{ \pi\in \calP(n): \delta^-_{\pi_{<1}}(1) \leq \lfloor\Delta/2 \rfloor = \delta^-_{\pi_{<2}}(2) ~\wedge~ 1\in \pi_{<2} \right\} \text{ and} \\
        \calP_{2<1}(n) & := \left\{ \pi\in \calP(n): \delta^-_{\pi_{<1}}(1) \leq \lfloor\Delta/2 \rfloor = \delta^-_{\pi_{<2}}(2) ~\wedge~ 2\in \pi_{<1} \right\},
    \end{align*}
    and observe that since~$\pi$ is chosen uniformly at random it is enough to show that $|\calP_{1<2}(n)| \geq |\calP_{2<1}(n)|$. We do so by constructing an injective function $g:\calP_{2<1}(n) \to \calP_{1<2}(n)$.

    For each $\pi\in\calP(n)$, let $i(1),i(2)\in \{1,\ldots,n\}$ be such that $\pi_{i(1)} = 1$ and $\pi_{i(2)} = 2$, \ie $i(j)$ is the position of vertex $j$ in the permutation $\pi$, for $j\in \{1,2\}$. Now define $g$ such that for all $\pi \in \calP_{2<1}(n)$, $g(\pi) = \pi^{i(1),i(2)}$.
    This function is injective: by definition of $\pi^{i(1),i(2)}$, $g(\pi)=g(\pi')$ implies $\pi=\pi'$. To show that the codomain of $g$ is $\calP_{1<2}(n)$, consider $\pi\in\calP_{2<1}(n)$.
    Since $2\in\pi_{<1}$, it follows that $1\in(g(\pi))_{<2}$.
    Moreover $(g(\pi))_{<1} = \pi_{<1} \cap \pi_{<2} \subseteq \pi_{<1}$ and thus $\delta^-_{(g(\pi))_{<1}}(1) \leq \delta^-_{\pi_{<1}}(1) \leq \lfloor \Delta/2\rfloor$.
    Similarly, $(g(\pi))_{<2} = \pi_{<2} \cup \pi_{<1} \cup \{1\} \supset \pi_{<2}$ and thus $\delta^-_{(g(\pi))_{<2}}(2) \geq \delta^-_{\pi_{<2}}(2) = \lfloor \Delta/2\rfloor$. Since $\delta^-(2)=\lfloor\Delta/2\rfloor$ the last inequality must hold with equality, which implies that $g(\pi)\in \calP_{1<2}(n)$. This shows \eqref{claim:cond-prob-first-vertex} and completes the proof of the theorem.    

\section{Proof of Lemma \ref{lem:random-dict}}
\label{app:lem-random-dict}

The mechanism is trivially impartial on $\calG$, since for any graph $G=(V,E)\in \calG$, we have $\mathsf{RD}_v(G) = \delta^-(v, G)/n$ for every $v\in V$, and $\delta^-(v, G) = \delta^-(v, G')$ for every $v\in V$ and $G'$ with $G_{-v} = G'_{-v}$.
    To see that the mechanism is $1/2+1/n$-optimal on $G_n$ for $n\in \{2,\ldots,5\}$, let $n\in \NN,~G\in \calG_n$, and observe that 
    \[
        \sum_{v\in V} \delta^-(v) \mathsf{RD}_v(G) = \frac{1}{n} \sum_{v\in V} (\delta^-(v))^2 = \frac{1}{n} \left( \Delta^2 + \sum_{v\in V\setminus v^*}(\delta^-(v))^2 \right).
    \]
    Therefore, the mechanism is $\alpha_n$-optimal on $\calG_n$ for 
    \[
        \alpha_n = \min_{\substack{G\in \calG_n:\\ \Delta(G)>0}} \frac{1}{n}\left( \Delta(G) + \frac{1}{\Delta(G)}\sum_{v\in V\setminus v^*(G)}(\delta^-(v,G))^2 \right) = \min_{\Delta\in \{1,\ldots,n-1\}} \frac{1}{n}\left( \Delta + \frac{n-\Delta}{\Delta} \right),
    \]
    where we observed that the minimum is reached on graphs where all but one vertex have indegree $0$ or $1$.
    For constant $n$, the expression on the right is convex in $\Delta$ and reaches its minimum for $\Delta=\sqrt{n}$. 
    For $n\in \{2,3,4,5\}$, it is easy to see that the minimum taken over natural values of $\Delta$ is reached for $\Delta=2$, which implies $\alpha_n = 1/2+1/n$ and concludes the proof of the lemma.

\section{Proof of Lemma \ref{lem:tight-instances-permutation}}
\label{app:lem-tight-instances-permutation}

Let $G=(V,E)\in G$ be an arbitrary graph.
    If $\Delta\not \in \{2,3\}$, we know from \autoref{thm:lb-permutation} that the permutation mechanism is $\alpha$-optimal on $G$ with $\alpha\geq 7/10 > 31/45$, so we assume $\Delta\in \{2,3\}$ in what follows.
    For $v\in T(G)$, we let $\smash{b_v=\delta^-_{T(G)}(v)}$ denote the indegree of each vertex in $T(G)$ from other maximum-indegree vertices and $v^R$ the right-most vertex among this subset in the permutation $\pi$ taken uniformly at random, \ie $v^R\in T(G)$ such that $v\in \pi_{<v^R}$ for every $v\in T(G)\setminus\{v^R\}$.
    Then,
    \begin{align*}
        \EE\left[\delta^-_{\pi_{<v^R}}(v^R)\right] & = \frac{1}{|T(G)|}\sum_{v\in T(G)} \left(b_v + \sum_{u\in N^-(v)\setminus T(G)} \PP[u\in \pi_{<v}] \right)\\
        & = \frac{1}{|T(G)|}\sum_{v\in T(G)} \left(b_v + \frac{|T(G)|}{|T(G)|+1}\left(\Delta-b_v\right) \right)\\
        & = \frac{|T(G)|}{|T(G)|+1}\Delta + \frac{1}{|T(G)|(|T(G)|-1)} \sum_{v\in T(G)}b_v,
    \end{align*}
    where we used that the probability of any vertex not in $T(G)$ being before the last vertex of $T(G)$ in the permutation is $|T(G)|/(|T(G)|+1)$.
    If $|T(G)|\geq 3$, this yields $\EE\left[\delta^-_{\pi_{<v^R}}(v^R)\right] / \Delta \geq 3/4$ and thus $3/4$-optimality of the permutation mechanism due to \autoref{lem:max-indegree-left}.
    If  $|T(G)|=2$ and $b_v=1$ for some $v\in T(G)$, we get  $\EE\left[\delta^-_{\pi_{<v^R}}(v^R)\right] / \Delta \geq 3/4>31/45$ if $\Delta=2$, and $\EE\left[\delta^-_{\pi_{<v^R}}(v^R)\right] / \Delta \geq 13/18>31/45$ if $\Delta=3$, implying these values as approximation guarantees for the permutation mechanism.
    If $\Delta=2,~ |T(G)|=2$, and $b_v=0$ for both vertices $v\in T(G)$, the indegrees of these vertices from the left are independent random variables.
    Denoting $T(G)=\{v^*_1,v^*_2\}$ and borrowing some notation from \autoref{sec:lower-bound} to denote as $X$ the random variable corresponding to the indegree of the vertex selected by the permutation mechanism, this implies
    \begin{align*}
        \frac{1}{2}\EE[X] & \geq \PP\left[ \max_{v\in T(G)}\delta^-_{\pi_{<v}}(v)=2\right] + \frac{1}{2}\PP\left[ \max_{v\in T(G)}\delta^-_{\pi_{<v}}(v)\leq 1\right]\\
        &  =1 - \frac{1}{2}\PP\left[ \delta^-_{\pi_{<v^*_1}}(v^*_1) \leq 1 \wedge \delta^-_{\pi_{<v^*_2}}(v^*_2) \leq 1  \right] = 1- \frac{1}{2}\left(\frac{2}{3}\right)^2 = \frac{7}{9} > \frac{31}{45},
    \end{align*}
    where the first equality uses \autoref{lem:max-indegree-left} to ensure that a vertex with strictly positive indegree is always selected.
    Similarly, if $\Delta=3,~ T(G)=\{v^*_1,v^*_2\}$, $b_v=0$ for both vertices $v\in T(G)$ and $\delta^-(v)\leq 1$ for every other vertex $v\not\in T(G)$, we obtain
    \begin{align*}
        \frac{1}{3}\EE[X] & \geq \PP\left[ \max_{v\in T(G)}\delta^-_{\pi_{<v}} \geq 2\right] +  \frac{1}{3} \PP\left[ \max_{v\in T(G)}\delta^-_{\pi_{<v}}\leq 1\right]\\
        &  \geq 1 - \frac{2}{3}\PP\left[ \delta^-_{\pi_{<v^*_1}}(v^*_1) \leq 1 \wedge \delta^-_{\pi_{<v^*_2}}(v^*_2) \leq 1  \right] = 1-\frac{2}{3} \left(\frac{1}{2}\right)^2 = \frac{5}{6} > \frac{31}{45}.
    \end{align*}

    We finally analyze the case when $\Delta=3$ and there exists $\bar{v}\in V$ with $\delta^-(\bar{v})=2$.
    We know from the proof of \autoref{thm:lb-permutation} that
    \begin{equation}\label{eq:bound-EE-X-Bi-Ai}
        \EE[X] \geq \frac{1}{\Delta+1}\left( \Delta^2 - \sum_{i=1}^{\Delta}(\Delta-2i)\PP\left[ B_i ~|~ A_i \right] \right) = \frac{1}{4}\left( 9- \sum_{i=1}^{3}(3-2i)\PP\left[ B_i ~|~ A_i \right] \right),
    \end{equation}
    where, for each $i\in \{0,\ldots,3\}$, $A_i = \left[\delta^-_{\pi<{v^*}}(v^*) = i\right]$ and $B_i = \bigcup_{v\in V\setminus \{v^*\}}\left[\delta^-_{\pi_{<v}}(v)\geq i \right]$.
    We claim that 
    \[
        \PP\left[\delta^-_{\bar{v}}(\bar{v}) = 2 ~|~ A_2\right] = \left\{ \begin{array}{ll}
             1/3 &  \text{ if } E \cap \{(\bar{v},v^*),(v^*,\bar{v})\} = \emptyset, \\
             3/10 &  \text{ if } E \cap \{(\bar{v},v^*),(v^*,\bar{v})\} = \{(v^*,\bar{v})\}, \\
             14/45 &  \text{ if } E \cap \{(\bar{v},v^*),(v^*,\bar{v})\} = \{(\bar{v},v^*)\}, \\
             4/15 &  \text{ otherwise.}
             \end{array}
   \right.
    \]
    The case $E \cap \{(\bar{v},v^*),(v^*,\bar{v})\} = \emptyset$ is straighforward, since in such situation the indegrees from the left of $v^*$ and $\bar{v}$ are independent. 
    For the other cases, we first fix the permutation $\pi(\{v^*\} \cup N^-(v^*))$ since conditioning on $B_2$ is equivalent to conditioning on $v^*$ being on the third position of this permutation.
    If $E \cap \{(\bar{v},v^*),(v^*,\bar{v})\} = \{(v^*,\bar{v})\}$, in order to have $\delta^-_{\bar{v}}(\bar{v}) = 2$ we first need that $\bar{v}$ is assigned to a position in the permutation either between $v^*$ and $\pi_4(\{v^*\} \cup N^-(v^*))$ or after $\pi_4(\{v^*\} \cup N^-(v^*))$, both events having probability $1/5$.
    We then need that the other inneighbor of $\bar{v}$, which we denote as $u$, is assigned before $\bar{v}$, which happens with probability $4/6$ if $\bar{v}$ was assigned between $v^*$ and $\pi_4(\{v^*\} \cup N^-(v^*))$ and with probability $5/6$ if it was assigned after $\pi_4(\{v^*\} \cup N^-(v^*))$.
    Putting all together, we obtain
    \begin{align*}
        \PP\left[\delta^-_{\bar{v}}(\bar{v}) = 2 ~|~ A_2\right] & = \PP\left[ u\in \pi_{<\bar{v}} \wedge \pi_3(\{v^*\} \cup N^-(v^*)) \in \pi_{<\bar{v}} \right]\\
        & = \PP\left[\pi_3(\{v^*\} \cup N^-(v^*)) \in \pi_{<\bar{v}}\right] \PP\left[ u\in \pi_{<\bar{v}} ~|~ \pi_3(\{v^*\} \cup N^-(v^*)) \in \pi_{<\bar{v}} \right]\\
        & = \frac{1}{5}\cdot \frac{4}{6} + \frac{1}{5}\cdot \frac{5}{6} = \frac{3}{10},
    \end{align*}
    where for the last equality we split the event $[\pi_3(\{v^*\} \cup N^-(v^*)) \in \pi_{<\bar{v}}]$ into the disjoint events
    \[
        \left[\pi_i(\{v^*\} \cup N^-(v^*)) \in \pi_{<\bar{v}} \text{ and } \bar{v} \in \pi_{<\pi_{i+1}(\{v^*\} \cup N^-(v^*))}\right]
    \]
    for $i\in \{3,4\}$, and computed the probability $\PP\left[ u\in \pi_{<\bar{v}}\right]$ conditioning on this events.
    If $E \cap \{(\bar{v},v^*),(v^*,\bar{v})\} = \{(\bar{v},v^*)\}$, in order to have $\delta^-_{\bar{v}}(\bar{v}) = 2$ we can have either $\bar{v}=\pi_1(\{v^*\} \cup N^-(v^*)),~ \bar{v}=\pi_2(\{v^*\} \cup N^-(v^*))$, or $\bar{v}=\pi_4(\{v^*\} \cup N^-(v^*))$, all these events with a probability of $1/3$.
    It then holds $\delta^-_{\bar{v}}(\bar{v}) = 2$ if both in-neighbors of $\bar{v}$, which we denote $u_1$ and $u_2$, are assigned before $\bar{v}$ in the permutation, which occurs with probability $1/5 \cdot 2/6$ in the first case, $2/5 \cdot 3/6$ in the second case, and $4/5 \cdot 5/6$ in the third case.
    We obtain
    \begin{align*}
        \PP\left[\delta^-_{\bar{v}}(\bar{v}) = 2 ~|~ A_2\right] & = \frac{1}{3} \sum_{i\in \{1,2,4\}} \PP\left[ u_1\in \pi_{<\pi_i(\{v^*\} \cup N^-(v^*))} \wedge u_2\in \pi_{<\pi_i(\{v^*\} \cup N^-(v^*))} \right]\\
        & = \frac{1}{3} \sum_{i\in \{1,2,4\}} \frac{i}{5} \cdot \frac{i+1}{6} = \frac{1}{3} \left(\frac{1}{5} \cdot \frac{2}{6} + \frac{2}{5} \cdot \frac{3}{6} + \frac{4}{5} \cdot \frac{5}{6}\right) = \frac{1}{3}\cdot \frac{14}{15} = \frac{14}{45}.
    \end{align*}
    Finally, if $\{(\bar{v},v^*),(v^*,\bar{v})\} \subseteq E$ then $\delta^-_{\bar{v}}(\bar{v}) = 2$ if and only if $\bar{v}=\pi_4(\{v^*\} \cup N^-(v^*))$, which occurs with probability $1/3$ given $v^* = \pi_3(\{v^*\} \cup N^-(v^*))$, and $\bar{v}$'s other in-neighbor, which we denote again as $u$, is assigned before $\bar{v}$ in the permutation, which occurs with probability $4/5$ conditioned on the former event.
    This yields
    \begin{align*}
       \PP\left[\delta^-_{\bar{v}}(\bar{v}) = 2 ~|~ A_2\right] & = \PP\left[ \bar{v}=\pi_4(\{v^*\} \cup N^-(v^*)) \wedge u\in \pi_{<\bar{v}} ~|~ A_2 \right] \\
       & = \PP\left[ \bar{v}=\pi_4(\{v^*\} \cup N^-(v^*)) ~|~ A_2 \right] \PP\left[ u\in \pi_{<\bar{v}} ~|~ \bar{v}=\pi_4(\{v^*\} \cup N^-(v^*)) \right]\\
       & = \frac{1}{3} \cdot \frac{4}{5} = \frac{4}{15}.
    \end{align*}
    This concludes the proof of the claim, which overall implies $\PP[B_2 ~|~ A_2] \geq \PP\left[\delta^-_{\bar{v}}(\bar{v}) = 2 ~|~ A_2\right] \geq 4/15$.
    Replacing this in \eqref{eq:bound-EE-X-Bi-Ai} yields
    \[
        \EE[X] \geq \frac{1}{4} \left( 9- \PP[B_1 ~|~ A_1] + \PP[B_2 ~|~ A_2] \right) \geq \frac{1}{4} \left( 9- 1 + \frac{4}{15} \right) = \frac{31}{15},
    \]
    so the permutation mechanism is $31/45$-optimal on $G$.
    This concludes the proof of the lemma.
    
\section{Proof of Lemma \ref{lem:prugd}}
\label{app:lem-prugd}

For a graph $G\in \calG_n$ and a permutation $\pi\in \calP(n)$,
we let $p(\pi, G)$ denote the vector $p(\pi)$ as defined in \autoref{alg:pwru-gap} with $G$ as input graph, and we further denote as $v^F(\pi,G)$ and $v^S(\pi,G)$ the vertices $v^F(\pi)$ and $v^S(\pi)$ defined in \autoref{alg:pwru-gap} with $G$ as input graph.
As usual, we omit the dependence on $G$ when the graph is clear from the context.
Let $C(\pi,G)$ denote the event $$[\delta^-(v^F(\pi),G) \geq \delta^-(v,G_{-v^F(\pi)})+2 \text{ for every } v\in V\setminus \{v^F(\pi)\}],$$ and $C(G)=[C(\pi,G) \text{ holds for every } \pi\in \calP(n)]$.
The following lemma states the natural property that $C(\pi,G)$ is independent of $\pi$ and characterizes the probabilities assigned by the mechanism with input $G$ when the event holds for this graph.

\begin{lemma}
\label{lem:prug-gap}
    Let $n\in \NN,~ G\in \calG_n$ be arbitrary.
    Then, $C(\pi,G)$ holds for some $\pi\in \calP(n)$ if and only if $C(G)$ holds.
    Moreover, if $C(G)$ holds, then $T(G)=\{v^*\},~ |\{v\in V:\delta^-(v)=\Delta-1\}|\leq 1$, and for $\pi\in \calP(n)$, it holds $p_{v^*}(\pi,G)=3/4$ and
    \[
        p_{v^S(\pi,G)}(\pi,G) = \left\{ \begin{array}{ll} 
        1/2 & \text{if } \delta^-\left(v^S(\pi,G)\right) =\Delta-1,~ \left(v^S(\pi,G),v^*\right)\in E \text{ and } \pi_{v^S(\pi,G)} > \pi_{v^*}, \\[1ex]
        0 & \text{otherwise.}
        \end{array} \right.
    \]
\end{lemma}

\begin{proof}
    Let $n,~ G=(V,E)$ be as in the statement of the lemma, and let $\bar{\pi}\in \calP(n)$ be such that $\delta^-(v^F(\bar{\pi}),G) \geq \delta^-(v,G_{-v^F(\bar{\pi})})+2$ for every $v\in V\setminus \{v^F(\bar{\pi})\}$.
    This implies $\delta^-(v^F(\bar{\pi}),G)\geq \delta^-(v,G)+1$ for every $v\in V\setminus \{v^F(\bar{\pi})\}$, which yields $T(G)=\{v^*\}$ and $v^F(\pi)=v^*$ for every $\pi\in \calP(n)$.
    Furthermore, since $\delta^-(v^*,G) \geq \delta^-(v,G_{-v^F(\bar{\pi})})+2$ for every $v\in V\setminus \{v^F(\bar{\pi})\}$, if there is $v\in V$ with $\delta^-(v,G)=\Delta(G)-1$ it necessarily holds $v\in N^-(v^*,G)$, which implies both $|\{v\in V:\delta^-(v,G)=\Delta(G)-1\}|\leq 1$ and $\delta^-(v^F(\pi),G) \geq \delta^-(v,G_{-v^F(\pi)})+2$ for every $\pi\in \calP(n),~ v\in V\setminus \{v^F(\pi)\}$.
    Therefore, $C(\pi,G)$ holds for every $\pi\in \calP(n)$, \ie $C(G)$ holds.
    This directly implies $p_{v^*}(\pi,G)=3/4$ for every $\pi\in \calP(n)$.
    Since every vertex other than $v^*$ has indegree at most $\Delta(G)-1$, in order to have $p_{v^S(\pi,G)}(\pi,G)=1/2$, from \autoref{alg:pwru-gap} we know that we need to have $\delta^-(v^S(\pi,G),G)=\Delta(G)-1,~ (v^S(\pi,G),v^*)\in E$, and $\pi_{v^S(\pi,G)} > \pi_{v^*}$.
    Otherwise, $p_{v^S(\pi,G)}(\pi,G)=0$.
    This concludes the proof of the lemma.
\end{proof}

\autoref{lem:prug-gap} suggests a simple alternative way of computing the probabilities assigned to each vertex when running plurality with runner-up and gap, by replacing probabilities $p_v(\pi,G)=1/2,~ p_v(\pi^R,G)=0$ by $p'_v(\pi,G)=p'_v(\pi^R,G)=1/4$.
This will later simplify the proof for the performance guaranteed by the mechanism.
For $G\in \calG_n,~ \pi\in \calP(n)$, consider $p'(\pi,G)$ defined for every $v\in V$ as
\[
    p'_v(\pi,G) = \left\{ \begin{array}{ll} 
        1/4 & \text{if } C(G) \text{ holds }, \delta^-(v) =\Delta-1, \text{ and } (v,v^*)\in E, \\[1ex]
        p(\pi,G) & \text{otherwise.}
        \end{array} \right.
\]
The following lemma states both that the probability assigned by the mechanism with input graph $G$ to any vertex $v$ is the average of $p_v(\pi,G)$ over $\pi$ and that this same average can be computed using $p'_v(\pi,G)$ as well.

\begin{lemma}
\label{lem:equivalent-probs-prug}
    For every $n\in \NN,~ G\in \calG_n$, it holds
    \[
        \mathsf{PRUG}(G) = \frac{1}{n!} \sum_{\pi\in \calP(n)} p(\pi,G) = \frac{1}{n!} \sum_{\pi\in \calP(n)} p'(\pi,G).
    \]
\end{lemma}

\begin{proof}
    For a graph $G\in \calG_n$ and a permutation $\pi\in \calP(n)$, let $q(\pi, G)$ denote the vector $q(\pi)$ as defined in \autoref{alg:pwru-gap} with $G$ as input graph when the permutation sampled in the first step is $\pi$, \ie $q=(p(\pi,G)+p(\pi^R,G))/2$.
    Let $n\in \NN$ and $G=(V,E) \in \calG_n$ be arbitrary.
    To see the first equality, note that
    \[
        \mathsf{PRUG}(G) = \frac{1}{n!} \sum_{\pi\in \calP(n)} q(\pi,G) = \frac{1}{n!} \sum_{\pi\in \calP(n)} \frac{1}{2}\left(p(\pi,G)+p\left(\pi^R,G\right)\right) = \frac{1}{n!} \sum_{\pi\in \calP(n)} p(\pi,G).
    \]
    The second equality follows immediately if $C(G)$ does not hold, or if there is no vertex $v\in V$ with $\delta^-(v)=\Delta-1$, or if $(v,v^*)\not\in E$ for every $v\in V$ with $\delta^-(v)=\Delta-1$.
    We assume, therefore, that $C(G)$ does hold, and that there is a vertex $v\in V$ with $\delta^-(v)=\Delta-1$ and $(v,v^*)\in E$.
    From \autoref{lem:prug-gap}, we know that this is in fact the unique vertex with indegree $\Delta-1$ and thus $v^S(\pi,G)=v$ for every $\pi\in \calP(n)$.
    Moreover, we know that $p_v(\pi,G)=1/2$ if $\pi_v > \pi_{v^*}$ and $p_v(\pi,G)=0$ otherwise, thus $(p(\pi,G)+p(\pi^R,G))/2=1/4= p'_v(\pi,G)$ for every $\pi\in \calP(n)$.
    We obtain
    \[
        \mathsf{PRUG}_v(G) = \frac{1}{n!} \sum_{\pi\in \calP(n)} q(\pi,G) = \frac{1}{n!} \sum_{\pi\in \calP(n)} \frac{1}{2}\left(p(\pi,G)+p\left(\pi^R,G\right)\right) = \frac{1}{n!} \sum_{\pi\in \calP(n)} p'_v(\pi,G),
    \]
    which yields the equality for vertex $v$.
    For every $u\in V\setminus \{v\}$, we know from the definition of $p'(\pi,G)$ that $p'_u(\pi,G)=p_u(\pi,G)$ for every $\pi\in \calP(n)$ and thus the equality holds directly.
    This concludes the proof of the lemma.  
\end{proof}

This last lemma directly implies that plurality with runner-up and gap is indeed an inexact mechanism.
We now state this property together with impartiality.
\begin{lemma}
\label{lem:prug}
    Plurality with runner-up and gap is an impartial inexact mechanism on $\calG$.
\end{lemma}
\begin{proof}
    Let $q(\pi,G)$ denote the vector $q$ defined in \autoref{alg:pwru-gap} when the input graph is $G$ and the sampled permutation is $\pi$, \ie $q(\pi,G)=(p(\pi,G)+p(\pi^R,G))/2$.
    We first show that plurality with runner-up and gap is an inexact mechanism on $\calG$, \ie that for every $G\in \calG$ it holds $\sum_{v\in V} \mathsf{PRUG}_v(G) \leq 1$.
    Let $n\in \NN$ and $G=(V,E)\in \calG_n$ be arbitrary.
    From \autoref{lem:prug-gap}, if $C(G)$ does not hold, then for every $\pi\in \calP(n)$ we have $p_{v^F(\pi)}(\pi)=1/2$ and thus $p_{v^F(\pi)}(\pi) + p_{v^S(\pi)}(\pi) \in \{1/2, 1\}, ~p_{v}(\pi)=0$ for every $v\not\in \{v^F(\pi), v^S(\pi)\}$. This directly implies $\sum_{v\in V}p_v(\pi)\leq 1$ for every $\pi\in \calP(n)$ and the result follows directly from \autoref{lem:equivalent-probs-prug}
    If $C(G)$ does hold, then \autoref{lem:equivalent-probs-prug} implies
    \begin{align*}
         \sum_{v\in V} \mathsf{PRUG}^D_v(G) & = \frac{1}{n!} \sum_{\pi\in \calP(n)} \sum_{v\in V} p_v(\pi,G)\\
         & = \frac{1}{n!} \sum_{\pi\in \calP(n)} [p_{v^F(\pi,G)}(\pi,G) + p_{v^S(\pi,G)}(\pi,G) ] \\
         & = \frac{1}{n!} \sum_{\pi\in \calP(n)} \left[\frac{3}{4} + \frac{1}{2}\chi( \delta^-(v^S(\pi,G)) =\Delta-1,~ (v^S,v^*)\in E,~ \pi_{v^S(\pi,G)} > \pi_{v^*} \right]\\
         & \leq \frac{1}{n!} \left( \frac{3}{4}n! + \frac{1}{4}n!\right) = 1,
    \end{align*}
    where the third equality follows from \autoref{lem:prug-gap} and the inequality follows since two vertices are in a certain order for half of the permutations.
    This concludes the proof of $\mathsf{PRUG}$ being an inexact mechanism.

    We now show impartiality.
    Let $G=(V,E),~ G'=(V,E')\in \calG$ and $v\in V$ such that $G_{-v} = G'_{-v}$.
    We show in the following that for every $\pi \in \calP(n)$ it holds $p_v(\pi,G) = p_v(\pi,G')$.
    Since 
    \[
        \mathsf{PRUG}(G) = \frac{1}{n!} \sum_{\pi\in \calP(n)} p(\pi,G),
    \]
    this implies $\mathsf{PRUG}_v(G)=\mathsf{PRUG}_v(G')$ and thus impartiality of the mechanism.
    Let $\pi\in \calP(n)$ and suppose first $p_v(\pi,G)=3/4$.
    From \autoref{lem:prug-gap}, we know that $v=v^*$, and $\delta^-(v^*,G) \geq \delta^-(v,G_{-v^*})+2$ for every $v\in V\setminus \{v^*\}$.
    Since $G_{-v^*} = G'_{-v^*}$, this implies $\delta^-(v^*,G') \geq \delta^-(v,G'_{-v^*})+2$ for every $v\in V\setminus \{v^*\}$, and thus $p_v(\pi,G')=3/4$.
    Suppose now $p_v(\pi,G)=1/2$, so either (i) $v=v^F(\pi,G)$ or (ii) $v=v^S(\pi,G),~ (v^S(\pi,G),v^F(\pi,G))\in E$ and either $\delta^-(v^S(\pi,G)) = \Delta(G)$ or both $(\delta^-(v^S(\pi,G)) = \Delta(G)-1$ and $\pi_{v^S(\pi,G)} > \pi_{v^F(\pi,G)}$.
    If $v=v^F(\pi,G)$, we have that
    \begin{equation}
        v=\arg \max_{u\in V} (\delta^-(u,G), \pi_u).\label{eq:top-voted-prug}
    \end{equation}
    If $v=\arg \max_{u\in V} (\delta^-(u,G'), \pi_u)$ as well, then clearly $p_v(\pi,G')=1/2$ and we conclude.
    Otherwise, there is $v'\in V\setminus \{v\}$ such that $v'=\arg \max_{u\in V} (\delta^-(u,G'), \pi_u)$.
    But since $G_{-v}=G'_{-v}$, this yields $(v,v')\in E$ and either (a) $\delta^-(v',G')=\Delta(G)+1$ and $\pi_{v}>\pi_{v'}$, or (b) $\delta^-(v',G')=\Delta(G)$ and $\pi_{v'}>\pi_{v}$.
    From \eqref{eq:top-voted-prug} we have $\pi_v>\pi_u$ for every $u\in V\setminus \{v,v'\}$ with $\delta^-(u,G')\geq \Delta(G)$, so we conclude $v=v^S(\pi,G')$ and thus $p_v(\pi,G')=1/2$.
    Consider now the case $v=v^S(\pi,G),~ (v^S(\pi,G),v^F(\pi,G))\in E$ and either $\delta^-(v^S(\pi,G)) = \Delta(G)$ or both $(\delta^-(v^S(\pi,G)) = \Delta(G)-1$ and $\pi_{v^S(\pi,G)} > \pi_{v^F(\pi,G)}$.
    If $v=v^S(\pi,G')$, then necessarily $(v,v^F(\pi,G'))\in E$: otherwise we would have $(\delta^-(v,G'),\pi_v) > (\delta^-(v^F(\pi,G'),G'),\pi_{v^F(\pi,G')})$, a contradiction.
    Since the other conditions do not change, we obtain $p_v(\pi,G')=1/2$.
    Otherwise, we necessarily have $v=v^F(\pi,G')$: if there were $u_1,u_2\in V$ with $(\delta^-(u_i,G'),\pi_{u_i})>(\delta^-(v,G'),\pi_v)$ for $i\in \{1,2\}$, we would have $(v^*,u_i)$ for $i\in \{1,2\}$, a contradiction.
    We conclude $p_v(\pi,G')=1/2$ as well, and thus $p_v(\pi,G) = p_v(\pi,G')$: the mechanism is impartial.
\end{proof}

Recall that, given an inexact mechanism $f$ and a graph $G$, we define the addition of a default vertex as the procedure that selects a vertex~$\bar{v}$ uniformly at random, removes its outgoing edge, and runs~$f$ on the resulting graph; if~$f$ ends up not selecting a vertex, the new mechanism selects vertex~$\bar{v}$. This is described in \autoref{alg:default-vertex} and its output is denoted as $\mathsf{DV}(f,G)$ for an inexact mechanism~$f$ and a graph~$G$.
The following lemma states that the addition of a default vertex turns any inexact mechanism into a mechanism while preserving impartiality. 
\begin{lemma}
\label{lem:default}
    Let $f$ be an inexact mechanism that is impartial on $\calG_n$. Then $f^D:\calG_n\to [0,1]^n$ given by $f^D(G)=\mathsf{DV}(f,G)$ is an impartial mechanism on $\calG_n$.
\end{lemma}
\begin{proof}
    Exactness is straightforward from the definition of probabilities in \autoref{alg:default-vertex}.
    To check impartiality, let $n\in \NN,~ G=(V,E),\ G'=(V,E')\in \calG_n$, and $v\in V$ such that $G_{-v} = G'_{-v}$. 
    From the definition of addition of a default vertex with input mechanism $f$, we have that 
    \begin{align*}
        f^D_v(G) = \frac{1}{n} \sum_{\bar{v}\in V} f_v((V,E\setminus (\{\bar{v}\} \times V))) + \frac{1}{n}\left[ 1 - \sum_{u\in V} f_u((V,E\setminus (\{v\}\times V))) \right],\\
        f^D_v(G') = \frac{1}{n} \sum_{\bar{v}\in V} f_v((V,E'\setminus (\{\bar{v}\} \times V))) + \frac{1}{n}\left[ 1 - \sum_{u\in V} f_u((V,E'\setminus (\{v\}\times V))) \right],
    \end{align*}
    where each term of the first summation on the right-hand side of both expressions corresponds to the probability assigned to vertex $v$ by mechanism $f$ when vertex $\bar{v}$ is the default vertex and thus its outgoing edges are omitted, and the second term is the additional probability assigned to vertex $v$ when it is the default vertex.
    Since $(V,E\setminus (\{v\}\times V))$ and $(V,E'\setminus (\{v\}\times V))$ are the same graph, the second terms on each right-hand side are trivially equal.
    Moreover, $G_{-v} = G'_{-v}$ implies that for every $\bar{v}\in V$ it holds
    \[
        (E\setminus (\{\bar{v}\} \times V))\setminus (\{v\}\times V) = (E'\setminus (\{\bar{v}\} \times V))\setminus (\{v\}\times V),
    \]
    thus from impartiality of $f$ we obtain
    \[
        f_v((V,E\setminus (\{\bar{v}\} \times V))) = f_v((V,E'\setminus (\{\bar{v}\} \times V))).
    \]
    We conclude that $f^D_v(G) = f^D_v(G')$, so the mechanism is impartial.
\end{proof}

We have now all the necessary ingredients to prove \autoref{lem:prugd}, regarding the mechanism $\mathsf{PRUG}^D$ obtained from adding a default vertex to plurality with runner-up and gap.
Formally, for every $G\in\calG$ we define $\mathsf{PRUG}^D(G)=\mathsf{DV}(\mathsf{PRUG},G)$. 

\begin{proof}[Proof of \autoref{lem:prugd}]
    That $\mathsf{PRUG}^D$ is an impartial mechanism follows directly from \autoref{lem:prug} and \autoref{lem:default}.
    
    In order to study its approximation ratio, let $n\in \NN$ be any natural value with $n\geq 6$ and $G=(V,E)\in \calG_n$ an arbitrary graph with $n$ vertices and $\Delta\geq 2$ (otherwise the result is trivial). 
    From \autoref{lem:prug-gap}, we know that $C(G_{-\bar{v}})$ is equivalent to the existence of $\pi\in \calP(n)$ such that
    \[
        \delta^-\left(v^F(\pi,G_{-\bar{v}}),G_{-\bar{v}}\right) \geq \delta^-(v,(G_{-\bar{v}})_{-v^F(\pi,G_{-\bar{v}})})+2 \text{ for every } v\in V\setminus \left\{v^F(\pi,G_{-\bar{v}})\right\},
    \]
    and to the fact that this holds for every $\pi\in \calP(n)$.
    For each default vertex $\bar{v}\in V$, we define the event $D(\pi,\bar{v}) = \left[p'_{v^S(\pi,G_{-\bar{v}})}(\pi, G_{-\bar{v}}) > 0 \right]$ and define $v_2(\pi,\bar{v})$ as
    \[
        v_2(\pi,\bar{v}) = \left\{ \begin{array}{ll} 
        v^S(\pi,G_{-\bar{v}}) & \text{if } D(\pi,\bar{v}) \text{ holds,} \\[1ex]
        \bar{v} & \text{otherwise,}
        \end{array} \right.
    \]
    \ie the vertex that is assigned strictly positive probability according to $p'(\pi,G_{-\bar{v}})$ in addition to $v^F(\pi,G_{-\bar{v}})$ for the permutation $\pi$ when $\bar{v}$ is the default vertex.
    For $\bar{v}\in V$ and $\pi\in \calP(n)$, $v_2(\pi,\bar{v})$ is assigned probability $1/4$ if $C(G_{-\bar{v}})$ holds and $1/2$ otherwise, due to  the definitions of $p(\pi,G)$ and $p'(\pi,G)$.
    Moreover, 
    we can now compute the expected indegree of the vertex selected by the mechanism, $\sum_{v\in V} \delta^-(v)\mathsf{PRUG}^D_v(G)$, as
    \begin{align}
        \frac{1}{n} \sum_{\bar{v}\in V} &\left[ \sum_{v\in V}\delta^-(v) \mathsf{PRUG}_v(G_{-\bar{v}}) + \delta^-(\bar{v}) \left( 1-\sum_{v\in V} \mathsf{PRUG}_v(G_{-\bar{v}}) \right) \right] \nonumber \\
        & = \frac{1}{n\cdot n!} \sum_{\bar{v}\in V} \left[ \sum_{\pi\in \calP(n)} \sum_{v\in V}\delta^-(v) p'_v(\pi, G_{-\bar{v}}) + \delta^-(\bar{v}) \left( 1- \sum_{\pi\in \calP(n)} \sum_{v\in V} p'_v(\pi, G_{-\bar{v}}) \right) \right]\nonumber \\
        & = \frac{1}{n\cdot n!} \sum_{\bar{v}\in V:~ C(G_{-\bar{v}})} \sum_{\pi\in \calP(n)} \biggl[ \frac{3}{4}\delta^-\left(v^F(\pi,G_{-\bar{v}})\right) + \frac{1}{4}\delta^-\left(v^S(\pi,G_{-\bar{v}})\right) \chi(D(\pi,\bar{v})) + {} \nonumber \\[-1.5ex]
        & \hspace*{8cm} \frac{1}{4}\delta^-(\bar{v}) \chi( \neg D(\pi,\bar{v})) \biggr] + {} \nonumber \\
        & \phantom{{}={}} \frac{1}{n\cdot n!} \sum_{\substack{\bar{v}\in V:\\ \neg C(G_{-\bar{v}})}} \sum_{\pi\in \calP(n)} \biggl[ \frac{1}{2}\delta^-\left(v^F(\pi,G_{-\bar{v}})\right) + \frac{1}{2}\delta^-\left(v^S(\pi,G_{-\bar{v}})\right) \chi(D(\pi,\bar{v})) + {} \nonumber \\[-3ex]
        &\hspace*{8cm} \frac{1}{2}\delta^-(\bar{v}) \chi( \neg D(\pi,\bar{v})) \biggr] \nonumber \\
        & = \frac{1}{2n\cdot n!} \sum_{\bar{v}\in V} \sum_{\pi\in \calP(n)} \left( \delta^-\left(v^F(\pi,G_{-\bar{v}})\right) + \delta^-(v_2(\pi,\bar{v})) \right) + {} \nonumber \\
        & \phantom{{}={}} \frac{1}{4n\cdot n!} \sum_{\bar{v}\in V:~ C(G_{-\bar{v}})}\sum_{\pi\in \calP(n)} \left( \delta^-\left(v^F(\pi,G_{-\bar{v}})\right) - \delta^-(v_2(\pi,\bar{v})) \right). \label{eq:expanded-expectation}
    \end{align}
    The first expression comes from the definition of the addition of a default vertex; the second one, from \autoref{lem:equivalent-probs-prug}; the third one, from \autoref{lem:prug-gap} and the definition of $p'(\pi,G)$; and the last one, from the definition of $v_2(\pi,\bar{v})$. 

   Let $k=|\{v\in V: \delta^-(v)=\Delta-1\}|$ denote the number of vertices with indegree $\Delta-1$. We claim the following:
    \begin{align}
        \frac{1}{n\cdot n!}\sum_{\bar{v}\in V} \sum_{\pi\in \calP(n)} \delta^-(v^F(\pi,G_{-\bar{v}})) & \geq \Delta - \frac{k}{k+1}\frac{\Delta}{n} \quad \text{if } |T(G)|=1,\label{eq:indegree-v1}\\
        \frac{1}{n\cdot n!}\sum_{\bar{v}\in V} \sum_{\pi\in \calP(n)} \delta^-(v_2(\pi, \bar{v})) & \geq \frac{n-1+\chi(C(G))}{n}. \label{eq:indegree-v2}
    \end{align}
    We first show \eqref{eq:indegree-v1}, so we assume $T(G)=\{v^*\}$.
    We observe that, if $v^F(\pi,G_{-\bar{v}})\not= v^*$ for some default vertex $\bar{v}\in V$ and some permutation $\pi\in \calP(n)$, then $\bar{v}\in N^-(v^*)$ and, moreover, $v^*\not \in \arg\max_{ \{v^*\} \cup \{v\in V: \delta^-(v)=\Delta-1\}} \pi_v$.
    These two events are independent; the former occurs for $\Delta$ out of $n$ possible realizations of the default vertex and the latter for $kn!/(k+1)$ out of $n!$ permutations. 
    We obtain
    \[
        \frac{1}{n\cdot n!}\sum_{\bar{v}\in V} \sum_{\pi\in \calP(n)} \delta^-(v^F(\pi,G_{-\bar{v}})) = \frac{1}{n\cdot n!} \sum_{\bar{v}\in V} \sum_{\pi\in \calP(n)} \left(\Delta -  \chi(v^F(\pi,G_{-\bar{v}}) \not= v^*) \right) \geq \Delta - \frac{k}{k+1}\cdot \frac{\Delta}{n}.
    \]

    To see \eqref{eq:indegree-v2}, we distinguish three cases.
    If $|T(G)|\geq 3$ and $v_2(\pi,\bar{v})\not=\bar{v}$ for some $\bar{v}\in V,~ \pi\in \calP(n)$, then this vertex is $v^S(\pi,G_{-\bar{v}})$ and since $\Delta(G_{-\bar{v}})=\Delta$ with $|T(G_{-\bar{v}})|\geq 2$ we obtain $\delta^-(v_2(\pi,\bar{v}))=\Delta\geq \delta^-(\bar{v})$.
    This yields
    \[
        \frac{1}{n\cdot n!}\sum_{\bar{v}\in V} \sum_{\pi\in \calP(n)} \delta^-(v_2(\pi, \bar{v})) \geq \frac{1}{n\cdot n!}\sum_{\bar{v}\in V} \sum_{\pi\in \calP(n)} \delta^-(\bar{v}) = 1.
    \]
    If $|T(G)|=2$ and $\delta^-(v_2(\pi,\bar{v})) < \delta^-(\bar{v})$ for some $\bar{v}\in V,~ \pi\in \calP(n)$, this vertex is $v^S(\pi,G_{-\bar{v}})$ as well. 
    Since $\Delta(G_{-\bar{v}})=\Delta$ and $|T(G_{-\bar{v}})|\geq 1$, we must have $\delta^-(v_2(\pi,\bar{v}))= \Delta-1,~ \delta^-(\bar{v})=\Delta$, and $\pi_{v_2(\pi,\bar{v})} > \pi_{\bar{v}}$.
    This yields
    \begin{align*}
        \frac{1}{n\cdot n!}\sum_{\bar{v}\in V} \sum_{\pi\in \calP(n)} \delta^-(v_2(\pi, \bar{v})) & \geq \frac{1}{n\cdot n!} \left[ \sum_{\bar{v}\in V} \sum_{\pi\in \calP(n)}  \delta^-(\bar{v}) - \sum_{\bar{v}\in T(G)} \sum_{\substack{\pi\in \calP(n):\\ \pi_{v_2(\pi,\bar{v})} > \pi_{\bar{v}}}} (\delta^-(v_2(\pi,\bar{v})) - \delta^-(\bar{v})) \right]\\
        & = 1 - \frac{1}{n\cdot n!} \cdot 2 \cdot \frac{n!}{2} = \frac{n-1}{n}.
    \end{align*}
    Finally, if $T(G)=\{v^*\}$ and $\delta^-(v_2(\pi,\bar{v})) < \delta^-(\bar{v})$ for some $\bar{v}\in V,~ \pi\in \calP(n)$, then $\bar{v}=v^*$.
    If this was not the case and $\delta^-(\bar{v})=\Delta-1$, then $|\{v\in V: \delta^-(v,G_{-\bar{v}}) \geq \Delta-1\}| \geq 2$ and thus $\delta^-(v_2(\pi,\bar{v}))\geq \Delta-1 = \delta^-(\bar{v})$ for any permutation $\pi$, a contradiction.
    If $\bar{v}\not=v^*$ and $\delta^-(\bar{v})< \Delta-1$, then $v_2(\pi,\bar{v}) \not = \bar{v}$ for some $\pi\in \calP(n)$ implies $\delta^-(v_2(\bar{v},\pi))\geq \Delta-2 \geq \delta^-(\bar{v})$, again a contradiction.
    Moreover, we have $\delta^-(v_2(\pi,v^*)) \geq \Delta-1$ for every $\pi\in \calP(n)$, since $\Delta(G_{-v^*})=\Delta$.
    We obtain
    \begin{align*}
        \frac{1}{n\cdot n!}\sum_{\bar{v}\in V} \sum_{\pi\in \calP(n)} \delta^-(v_2(\pi, \bar{v})) & = \frac{1}{n}\left(\frac{1}{n!} \sum_{\pi\in \calP(n)} \delta^-(v_2(\pi,v^*)) - \delta^-(v^*) \right) + \frac{1}{n\cdot n!}\sum_{\bar{v}\in V} \delta^-(\bar{v})\\
        & \geq \frac{1}{n}(\Delta-1 - \Delta) + 1 = \frac{n-1}{n}.
    \end{align*}
    Since $\delta^-(v_2(\pi,v^*)) \geq \Delta-1$ implies $v_2(\pi,v^*)=v^*$ when $\delta^-(v^*)-\delta^-(v)\geq 2$ for every $v\in V\setminus \{v^*\}$, it is clear for such case that $\delta^-(v_2(\pi,\bar{v})) < \delta^-(\bar{v})$ cannot hold for any $\bar{v}\in V,~ \pi\in \calP(n)$ and thus 
    \[
        \frac{1}{n\cdot n!}\sum_{\bar{v}\in V} \sum_{\pi\in \calP(n)} \delta^-(v_2(\pi, \bar{v})) \geq \frac{1}{n\cdot n!}\sum_{\bar{v}\in V} \sum_{\pi\in \calP(n)} \delta^-(\bar{v}) = 1.
    \]
    This concludes the proof of \eqref{eq:indegree-v2}.

    We now observe that, if $|T(G)|\geq 2$, for every $\bar{v}\in V$ we have $\neg C(G_{-\bar{v}})$ and, moreover, for every $\pi\in \calP(n)$ it holds $\max_{v\in V}\delta^-(v,G_{-\bar{v}}) = \Delta$.
    This implies
    \begin{align}
        \frac{1}{\Delta}\sum_{v\in V} \delta^-(v)\mathsf{PRUG}^D_v(G) & \geq \frac{1}{2n\cdot n!\cdot \Delta} \sum_{\bar{v}\in V} \sum_{\pi\in \calP(n)} \left( \delta^-\left(v^F(\pi,G_{-\bar{v}})\right) + \delta^-(v_2(\pi,\bar{v})) \right) \nonumber \\
        & = \frac{1}{2\Delta} \left(\Delta + \frac{n-1}{n}\right) \geq \frac{1}{2}+\frac{5}{12\Delta}, \label{eq:lb-two-top-voted}
    \end{align}
    where the equality follows from \eqref{eq:indegree-v2} and the last inequality from the fact $n\geq 6$.
    We make use of this fact to prove all three items of the lemma.

    We first prove \autoref{lem:prugd-i}.
    If $|T(G)|\geq 2$, from \eqref{eq:lb-two-top-voted} we have 
    \[
        \frac{1}{\Delta}\sum_{v\in V} \delta^-(v)\mathsf{PRUG}^D_v(G) \geq \frac{1}{2}+\frac{5}{12\Delta} > \frac{1}{2} + \frac{7\Delta - 9}{6\Delta(3\Delta-2)},
    \]
    where the last inequality holds since
    \[
        \frac{5}{12\Delta} > \frac{7\Delta - 9}{6\Delta(3\Delta-2)} \Longleftrightarrow 15\Delta - 10 > 14\Delta-18 \Longleftrightarrow \Delta > -8,
    \]
    so we reduce to the case $|T(G)|=1$ in what follows. 
    Replacing the bounds given by \eqref{eq:indegree-v1} and \eqref{eq:indegree-v2} in the expression \eqref{eq:expanded-expectation}, we obtain that $\sum_{v\in V} \delta^-(v)\mathsf{PRUG}^D_v(G)$ is bounded from below by
    \begin{equation}
        \frac{1}{2} \left( \Delta- \frac{k}{k+1}\cdot \frac{\Delta}{n} +\frac{n-1}{n} \right) + \frac{1}{4n\cdot n!} \sum_{\bar{v}\in V: ~C_(G_{-\bar{v}})} \sum_{\pi\in \calP(n)} \left( \delta^-(v^F(\pi,G_{-\bar{v}})) - \delta^-(v_2(\pi,\bar{v})) \right).\label{expanded-expectation-one-top-voted}
    \end{equation}
    If $k=1$, let $v'$ denote the unique vertex with $\delta^-(v')=\Delta-1$.
    For every $\bar{v}\in N^-(v')\setminus \{v^*\}$ (a set of size at least $\Delta-2$), it holds $C(G_{-\bar{v}})$ and, for every $\pi\in \calP(n),~ \delta^-(v^F(\pi,G_{-\bar{v}}))=\Delta$ and $\delta^-(v_2(\pi,\bar{v}))\leq \Delta-2$.
    This yields
    \[
        \frac{1}{\Delta}\sum_{v\in V} \delta^-(v)\mathsf{PRUG}^D_v(G) \geq \frac{1}{2\Delta} \left( \Delta- \frac{\Delta}{2n} +\frac{n-1}{n} \right) + \frac{1}{4n\Delta} \cdot 2(\Delta-2) = \frac{1}{2\Delta}\left(\Delta+1 +\frac{\Delta-6}{2n}\right).
    \]
    If $\Delta\leq 6$, this expression is non-decreasing in $n$ and it is minimized for $n=6$, where
    \[
        \frac{1}{\Delta}\sum_{v\in V} \delta^-(v)\mathsf{PRUG}^D_v(G) \geq \frac{13}{24} + \frac{1}{4\Delta} > \frac{1}{2} + \frac{7\Delta - 9}{6\Delta(3\Delta-2)},
    \]
    where the last inequality follows since
    \[
        \frac{1}{24} + \frac{1}{4\Delta} > \frac{7\Delta - 9}{6\Delta(3\Delta-2)} \Longleftrightarrow 3\Delta^2 + 16\Delta-12 > 28\Delta-36 \Longleftrightarrow \Delta^2-4\Delta+8 > 0,
    \]
    which clearly holds since the expression on the left is a convex quadratic function that is minimized for $\Delta=2$, where its value is $4$.
    If $\Delta\geq 6$, on the other hand, the expression is decreasing in $n$ and it is minimized for $n\to \infty$, where
    \[
        \frac{1}{\Delta}\sum_{v\in V} \delta^-(v)\mathsf{PRUG}^D_v(G) \geq \frac{1}{2} + \frac{1}{2\Delta} > \frac{1}{2} + \frac{7\Delta - 9}{6\Delta(3\Delta-2)},
    \]
    where the last inequality follows since
    \[
        \frac{1}{2\Delta} > \frac{7\Delta - 9}{6\Delta(3\Delta-2)} \Longleftrightarrow 9\Delta-6 > 7\Delta-9 \Longleftrightarrow \Delta> -\frac{3}{2}.
    \]
    To address the case $k\geq 2$ we first note that we can omit the second term in \eqref{expanded-expectation-one-top-voted} since whenever $C(G_{-\bar{v}})$ holds for some $\bar{v}\in V$ we have $\delta^-(v^F(\pi,G_{-\bar{v}}))=\Delta$ and thus the terms of the sum are non-negative.
    This yields
    \[
        \frac{1}{\Delta}\sum_{v\in V} \delta^-(v)\mathsf{PRUG}^D_v(G) \geq \frac{1}{2\Delta} \left( \Delta- \frac{k}{k+1}\cdot \frac{\Delta}{n} +\frac{n-1}{n} \right) = \frac{\Delta+1}{2\Delta} - \frac{k(\Delta+1)+1}{2n\Delta(k+1)}. 
    \]
    Since this expression is increasing in $n$ and we can bound $n$ from below with the sum of the known indegrees, $n\geq \Delta+k(\Delta-1)$, we obtain
    \begin{equation}
        \frac{1}{\Delta}\sum_{v\in V} \delta^-(v)\mathsf{PRUG}^D_v(G) \geq \frac{\Delta+1}{2\Delta} - \frac{k(\Delta+1)+1}{2\Delta(k+1)(\Delta+k(\Delta-1))}. \label{lb-k-Delta}
    \end{equation}
    Denote the substracted term as $g(k,\Delta)$. We note that it is decreasing in $k\geq 2$, since
    \[
        \frac{\partial g(k,\Delta)}{\partial k} = -\frac{(\Delta^2-1)k^2 + 2(\Delta-1)k - (\Delta^2-\Delta+1)}{2\Delta(k+1)^2(\Delta+k(\Delta-1))^2},
    \]
    where the numerator is a concave function with its greatest root in
    \[
        k^* = \frac{\sqrt{\Delta(\Delta-1)(\Delta^2+1)}}{(\Delta-1)(\Delta+1)} + \frac{1}{\Delta+1} \leq 2,
    \]
    where the last inequality holds since
    \begin{align*}
        \frac{\sqrt{\Delta(\Delta-1)(\Delta^2+1)}}{(\Delta-1)(\Delta+1)} \leq \frac{2\Delta+1}{\Delta+1} & \Longleftrightarrow \Delta(\Delta-1)(\Delta^2+1) \leq (2\Delta+1)^2(\Delta-1)^2 \\
        & \Longleftrightarrow 3\Delta^4-3\Delta^3-4\Delta^2+3\Delta+1 \geq 0,
    \end{align*}
    which holds for every $\Delta\geq 2$ since $-4\Delta^2\geq -2\Delta^3$ and $-5\Delta^3\geq -5/2\Delta^4$, thus the expression on the left is greater or equal to $1/2\Delta^4+3\Delta+1$, a value that is clearly non-negative.
    Therefore, we can bound the approximation ratio of the mechanism by replacing $k=2$ in \eqref{lb-k-Delta}.
    This yields
    \[
        \frac{1}{\Delta}\sum_{v\in V} \delta^-(v)\mathsf{PRUG}^D_v(G) \geq \frac{\Delta+1}{2\Delta} - \frac{2\Delta+3}{6\Delta(3\Delta-2)} = \frac{1}{2} + \frac{7\Delta - 9}{6\Delta(3\Delta-2)},
    \]
    which concludes the proof of \autoref{lem:prugd-i}.

    We now prove \autoref{lem:prugd-ii}.
    If $\Delta=2$ and $|T(G)|\geq 2$, \eqref{eq:lb-two-top-voted} yields
    \[
        \frac{1}{2}\sum_{v\in V} \delta^-(v)\mathsf{PRUG}^D_v(G) \geq \frac{17}{24} > \frac{65}{96},
    \]
    so once again we can reduce to the case $T(G)=\{v^*\}$.
    We first improve both bounds \eqref{eq:indegree-v1} and \eqref{eq:indegree-v2} for this case.
    As an improvement of the former, we claim that
    \begin{equation}
         \frac{1}{n\cdot n!}\sum_{\bar{v}\in V} \sum_{\pi\in \calP(n)} \delta^-(v^F(\pi,G_{-\bar{v}})) \geq 2 - \frac{n-3}{n-2}\cdot \frac{2}{n}.\label{eq:indegree-v1-Delta2}
    \end{equation}
    To see this, we note that if $v^F(\pi,G_{-\bar{v}})\not= v^*$ for some default vertex $\bar{v}\in V$ and some permutation $\pi\in \calP(n)$, then (a) $\bar{v}\in N^-(v^*)$; (b) $v^*\not \in \arg\max_{ v\in \{v^*\} \cup \{v\in V: \delta^-(v)=1\}} \pi_v$; and, denoting $N^+(v^*)=\{w\}$, we have (c)
    \[
        \neg \left( w = \arg\max_{ v\in \{v^*\} \cup \{v\in V: \delta^-(v)=1\} } \pi_v \text{ and } v^* = \arg\max_{ v\in \{v^*\} \cup \{v\in V: \delta^-(v)=1\}\setminus \{w\}} \pi_v \right).
    \]
    (a) occurs for $2$ out of $n$ realizations of the default vertex; (b) and (c) hold for
    \[
        1-\frac{1}{n-1}-\frac{1}{(n-1)(n-2)} = \frac{n-3}{n-2}
    \]
    out of $n!$ permutations, those with $v^*$ not having the highest index among the vertices in $\{v^*\} \cup \{v\in V: \delta^-(v)=\Delta-1\}$ nor the vertex in $N^+(v)$ having the highest index in this set and $v^*$ the second-to-highest. 
    Since the default vertex and the permutation are sampled independently, we obtain
    \[
        \frac{1}{n\cdot n!}\sum_{\bar{v}\in V} \sum_{\pi\in \calP(n)} \delta^-(v^F(\pi,G_{-\bar{v}})) = 2 - \frac{1}{n\cdot n!} \sum_{\bar{v}\in V} \sum_{\pi\in \calP(n)} \chi(v^F(\pi,G_{-\bar{v}}) \not= v^*) \geq 2 - \frac{n-3}{n-2}\cdot \frac{2}{n}.
    \]
    As an improvement of \eqref{eq:indegree-v2}, we claim that 
    \begin{equation}
        \frac{1}{n\cdot n!}\sum_{\bar{v}\in V} \sum_{\pi\in \calP(n)} \delta^-(v_2(\pi, \bar{v})) \geq 1-\frac{1}{n(n-2)}. \label{eq:indegree-v2-Delta2}
    \end{equation}
    It is clear, from the more general case analyzed before, that if $\delta^-(v_2(\pi, \bar{v})) < \delta^-(\bar{v})$ for some $\pi\in \calP(n)$ then $\bar{v}=v^*$.
    But the inequality $\delta^-(v_2(\pi,v^*)) < 2$ only holds if there exists $u\in V$ with $(u,v^*)\in E$ and $\pi_u>\pi_v$ for every $v\in V$ with $\delta^-(v,G_{-v^*})\geq 1$.
    This holds for at most $2n!/(n-2)$ permutations, those with a vertex in $N^-(v^*)$ (a set of size 2) in the first position among the vertices in $\{v^*\}\cup \{v\in V: \delta^-(v, G_{-v^*})\geq 1\}$ (a set of size at least $n-2$).
    We obtain that
    \[
        \delta^-(v_2(\pi,v^*))\geq 1 \text{ for every } \pi\in \calP(n),\quad |\{\pi\in \calP(n): \delta^-(v_2(\pi,v^*))\geq 2\}| \geq \frac{n-4}{n-2}n!.
    \]
    On the other hand, denoting as $v'\in V$ the vertex with $\delta^-(v',G)=0$, we have $\delta^-(v_2(\pi,v'))\geq 1$ for every $\pi\in \calP(n)$ such that 
    \[
        \arg\max_{v\in V\setminus v^T} (\delta^-(v), \pi_v) \in N^-(v^T),
    \]
    where $v^T = \arg\max_{v\in V} (\delta^-(v), \pi_v)$.
    If $N^+(v')=\{v^*\}$ this happens for at most $n!/(n-2)$ permutations: those where the second vertex with the highest index among $\{v\in V: \delta^-(v,G_{-v'})=1\}$ (a set of size $n-1$) is an in-neighbor of the first one.
    If $N^+(v')\not=\{v^*\}$, on the other hand, this happens for at most $n!/(n-3)$ permutations, those where the first vertex with the highest index among $\{v\in V: \delta^-(v,G_{-v'})=1\}$ (a set of size $n-3$) is an in-neighbor of $v^*$.
    We obtain that 
    \[
        |\{\pi\in \calP(n): \delta^-(v_2(\pi,v'))\geq 1\}| \geq \frac{1}{n-2}n!.
    \]
    Putting everything together, we obtain
    \begin{align*}
        \frac{1}{n\cdot n!}\sum_{\bar{v}\in V} \sum_{\pi\in \calP(n)} \delta^-(v_2(\pi, \bar{v})) & \geq  \frac{1}{n}\left(\frac{1}{n!} \sum_{\pi\in \calP(n)} (\delta^-(v_2(\pi,v^*)) + \delta^-(v_2(\pi,v'))) - \delta^-(v^*) - \delta^-(v') \right)\\
        &\quad + \frac{1}{n}\sum_{\bar{v}\in V} \delta^-(\bar{v})\\
        & \geq \frac{1}{n}\left(-\frac{2}{n-2} + \frac{1}{n-2} + n \right)= 1 - \frac{1}{n(n-2)},
    \end{align*}
    as claimed.
    Replacing the bounds obtained in \eqref{eq:indegree-v1-Delta2} and \eqref{eq:indegree-v2-Delta2} with $\Delta=2$ and $k=n-2$ in \eqref{eq:expanded-expectation}, we obtain
    \[
        \frac{1}{2}\sum_{v\in V} \delta^-(v)\mathsf{PRUG}^D_v(G) \geq \frac{1}{4}\left(2-\frac{n-3}{n-2}\cdot \frac{2}{n} \right) + \frac{1}{4n} \left( n-\frac{1}{n-2}\right) = \frac{3}{4}- \frac{2n-5}{4n(n-2)} \geq \frac{65}{96},
    \]
    where we used that the last term is decreasing in $n\geq 6$ since its derivative is
    \[
        \frac{n(n-2)-(2n-5)(n-1)}{2n(n-2)} = - \frac{n^2-5n+5}{2n(n-2)} < 0,
    \]
    and evaluated it in $n=6$.
    This concludes the proof for $\Delta=2$.
    
    We finally show \autoref{lem:prugd-iii}. If $\Delta=3$, and $|\{v\in V: \delta^-(v)\geq 2\}|=1$, then whenever $\bar{v}\not\in N^-(v^*)\cup \{v^*\}$ we have that $C(G_{-\bar{v}})$ holds, and for every $\pi\in \calP(n),~ \delta^-(v^F(\pi,G_{-\bar{v}}))=\Delta$, and $\delta^-(v_2(\pi,\bar{v}))\leq 1$. Therefore, from \eqref{eq:expanded-expectation} and \eqref{eq:indegree-v2},
    \[
       \frac{1}{3}\sum_{v\in V} \delta^-(v)\mathsf{PRUG}^D_v(G) \geq \frac{1}{6}( 3+ 1 ) + \frac{1}{12n}\cdot  2 (n-4) = \frac{5}{6}-\frac{2}{3n} \geq \frac{13}{18},
    \]
    where we used that the expression after the equality is increasing in $n$ and equals $13/18$ for $n=6$.
    This concludes the proof of the lemma.
\end{proof}

\bibliographystyle{abbrvnat}
\bibliography{CFK2023}

\end{document}